\documentclass[11pt,letter]{article}
\usepackage[utf8]{inputenc}
\usepackage{amsmath}
\usepackage{amssymb}
\usepackage{amsfonts}
\usepackage{amscd}
\usepackage{latexsym}
\usepackage{wasysym}
\usepackage{graphicx}
\usepackage{xcolor}
\usepackage{hyperref}
\usepackage[all,cmtip]{xy}
\usepackage{bbold}

\usepackage{amssymb}
\usepackage{amsthm}
\usepackage{enumitem}
\usepackage{mathtools,bbm}
\usepackage{mathtools,bbm}
\usepackage{xpatch}
\usepackage{version,xspace}

\usepackage[margin=1in]{geometry}
\usepackage{hyperref}%
\hypersetup{colorlinks=true, linkcolor=blue, breaklinks=true, urlcolor=blue}
\usepackage{url,doi}

\usepackage[utf8]{inputenc}
\usepackage{amsfonts}
\usepackage{amsthm}
\usepackage{enumitem}
\usepackage{mathtools,bbm}
\usepackage{mathtools,bbm}
\usepackage{xpatch}

\usepackage{lipsum}
\usepackage{mathtools}
\usepackage{cuted}
\usepackage{mathtools,bbm}

\usepackage[caption=false,font=footnotesize]{subfig}
\usepackage[font=small]{caption}
\graphicspath{{./images/}{./fig}}

\usepackage{enumitem}
\setlist{nosep}

\newcommand{\Co}{\Gamma}

\newcommand{\Nin}{\mathcal{N}_{\operatorname{in}}}
\newcommand{\Nout}{\mathcal{N}_{\operatorname{out}}}
\newcommand{\Nins}{\mathcal{N}_{\operatorname{in}}^s}

\newcommand{\until}[1]{\{1,\dots, #1\}}

\newtheorem{theorem}{Theorem}[section]
\newtheorem{corollary}[theorem]{Corollary}
\newtheorem{lemma}{Lemma}[section]
\newtheorem{proposition}[theorem]{Proposition}
\newtheorem{remark}[theorem]{Remark} 
\newtheorem{problem}[theorem]{Problem}

{
      \theoremstyle{plain}

}
\newtheorem{definition}{Definition}[section]

\DeclareSymbolFont{bbold}{U}{bbold}{m}{n}
\DeclareSymbolFontAlphabet{\mathbbold}{bbold}

\newcommand{\notto}{\overset{\operatorname{not}}{\to}}

\newcommand{\setdef}[2]{\{#1 \; | \; #2\}}

\newcommand{\din}{d_{\operatorname{in}}}

\newcommand{\dins}{d_{\operatorname{in}}^{\operatorname{s}}}

\usepackage{algpseudocode,algorithm,algorithmicx}

\usepackage{hyperref}%
\hypersetup{colorlinks=true, linkcolor=blue, breaklinks=true, urlcolor=blue}

\newcommand\oprocendsymbol{\hbox{$\square$}}
\newcommand\oprocend{\relax\ifmmode\else\unskip\hfill\fi\oprocendsymbol}

\DeclareSymbolFont{bbold}{U}{bbold}{m}{n}
\DeclareSymbolFontAlphabet{\mathbbold}{bbold}

\usepackage[prependcaption,colorinlistoftodos]{todonotes}

\newcommand\blfootnote[1]{%
  \begingroup
  \renewcommand\thefootnote{}\footnote{#1}%
  \addtocounter{footnote}{-1}%
  \endgroup
}

\title{A Network Formation Game for the Emergence of Hierarchies}

\author{Pedro Cisneros-Velarde \and Francesco Bullo}
\date{}

\begin{document}
\maketitle
\begin{abstract}
  We propose a novel network formation game that explains the emergence of
  various hierarchical structures in groups where self-interested or
  utility-maximizing individuals decide to establish or severe
  relationships of authority or collaboration among themselves.  We
  consider two settings: we first consider individuals who do not seek the other
  party's consent when establishing a relationship and then individuals who
  do. For both settings, we formally relate the emerged hierarchical
  structures with the novel inclusion of well-motivated hierarchy promoting
  terms in the individuals' utility functions.  We first analyze the game
  via a static analysis and characterize all the hierarchical structures
  that can be formed as its solutions.  We then consider the game played
  dynamically under stochastic interactions among individuals implementing
  better-response dynamics and analyze the nature of the converged networks.
%
\end{abstract}

\blfootnote{
  This work is supported by the U. S. Army Research Laboratory and the
  U. S. Army Research Office under grant number W911NF-15-1-0577. The
    views and conclusions contained in this document are those of the
    authors and should not be interpreted as representing the official
    policies, either expressed or implied, of the Army Research Laboratory
    or the U.S. Government.}
\blfootnote{Pedro Cisneros-Velarde (e-mail:
    pacisne@gmail.com) and Francesco Bullo (e-mail:
    bullo@ucsb.edu) are with the University of California, Santa Barbara.}
    
\section{Introduction}
Hierarchies are ubiquitous forms of organization. Not surprisingly, early comparative studies between administrative and economic hierarchies, and other studies of hierarchies in the the context of political organization, 
date back to the beginning of the twentieth century~\cite{AL:914,RM:1915}. 
Since then, their study has encompassed a large variety of disciplines like, for example,  
labor and organizational economics~\cite{AB-JW:01,SR:82,OEW:73}, moral and evolutionary psychology~\cite{MK-DL:83,TSR-APF:11,DDC:96,MVV-JT:15}, anthropology~\cite{BD:10} and sociology~\cite{PB:68,MS-FL-JP:12}. 
We remark that in the economics literature, the importance of hierarchies also comes from the study of firms: the structure of firms can be understood as hierarchical structures 
in production organizations (e.g., superiors having veto power over subordinates to access productive assets)~\cite{RHC:37,OH-JM:98}.
%
%
%
In general, the empirical study of the emergence of hierarchical structures in social networks has benefited from advancements in the field of sociometry~\cite{HAS:96}, and, theoretically, graph theoretical analysis has proven to be a useful tool for structural analysis of general hierarchical structures~\cite{DK:94}.
In this paper, we adopt a game theoretical perspective to study the emergence of hierarchical structures through what is known as \emph{network formation games}. 
We propose a 
game and analyze how its play leads to the emergence of hierarchies. The importance of our approach stems from the fact that this game theoretical study allows us to investigate the underlying structures in the individual utilities that promote the emergence of hierarchies from self-interested 
individuals. Then, the importance of this work is that it defines micro-mechanisms on the individual level that can lead to the emergence of hierarchical structures on the macro-level. 

To the best of our knowledge, there is no previous study 
about the emergence of hierarchical network structures from a network formation perspective. 
However, we briefly mention that there has been a considerable body of literature on the study of hierarchies 
from other game theoretical perspectives~\cite{MS-AvdN:01}. For example, there are studies in the literature 
about allocation rules that take into account explicit hierarchical relationships among the agents (e.g.,~\cite{MS-RPG-HN-ST:05}), and about the analysis of games where agents have veto power over the 
actions of other players as a result of an underlying hierarchical structure (e.g.,~\cite{RPG-GO-RvdB:92}). 
%
%

On the other hand, there is a mature body of work in the network formation game literature related to other applications. Generally speaking,  
in a network formation game, self-interested or utility-maximizing agents decide to form or severe links (interpreted as some sort of relationship) with other agents so that they can form network structures that maximize some utility or allocation value. The structure of connections resulting from playing the game induces an emergent network whose topology can have different interpretations according to the specific application the game is modeling. The literature on network formation games has a long trajectory which historically started from a cooperative perspective, and has later devised a huge boom from a non-cooperative perspective~\cite{MOJ:05,MS-AvdN:01}. 
We present an overview of the literature on network formation games according to a classification of different previous approaches, this classification also introduces relevant technical terms that will be useful in this paper. For other more algorithmic perspectives to the subject in terms of solutions concepts and efficiency, not mentioned in this review, we refer to the textbook treatment of the subject in~\cite[Chapter~19]{NN-TR-ET-VVV:07}. 
Network formation games have been analyzed via a \emph{static analysis}, that is, the analysis of solution concepts that describes the emergence of particular ``stable" networks as a result of the agents playing the game. Additionally, they have also been analyzed via a \emph{dynamic analysis}, that is, the study of the emergence of a subset (or possibly all) of the networks defined by the static analysis (also known as equilibrium selection) as the result of some defined dynamic rules that describe the evolution of the agents' actions~\cite{YS-MVDS:15,AW:01}. For example,~\cite{VB-SG:00} uses Nash equilibrium as solution concept, whereas~\cite{MJ-AW:02} uses the solution concept called \emph{pairwise stability} (first introduced in~\cite{MOJ-AW:96}); and both works establish a dynamical analysis of convergence towards ``stable" networks in their respective solution concept. The two solution concepts mentioned above are the ones that have been typically used in applications where the efficiency of the formed network was the main focus of study; however, in general, solution concepts are proposed (and sometimes renamed) depending on the specific application to study from the network formation perspective (whether there is only a static analysis or a dynamic one too). For example, other solution concepts exist 
like the generalization of pairwise stability called~\emph{strong stability} also studied in the context of network efficiency~\cite{MOJ-AvdN:05}, or like the solution concept introduced in~\cite{RJ-SM-JNT:06} which is 
particularly tailored for the problem of routing traffic on networks.

Broadly, the literature on network formation games can be classified according to the nature of the agents that play the game and/or the nature of the connections they can establish. 
%
We discuss first the possible nature of the agents. Traditionally, a first distinction is whether the agents are \emph{non-consensual} or~\emph{consensual}. 
A \emph{non-consensual} agent $i$ can establish links unilaterally with agent $j$ without $j$'s consent, i.e., without $j$ having to agree to accept such link (see, e.g.,~\cite{VB-SG:00}); and a~\emph{consensual} $i$ is the one who needs $j$'s consent or approval to establish a link (see, e.g.,~\cite{MJ-AW:02,RG-SS:04,EA-RJ-SM:09}).   
Another second major distinction regarding the agents is whether they all have the same structure in their utility or payoff function, i.e., the rewards or costs of making/maintaining/severing a link is the same for any agent, which is know as the \emph{homogeneous case}; or whether they have different structures in their utility function, which is the \emph{heterogeneous case}~\cite{AG-SG-JK:06}.
%

Regarding the distinction on the nature of the connections that agents can establish, network formation games can be classified depending on whether the formed network 
is \emph{undirected} (i.e., the edge or link between agents $i$ and $j$ has no direction) 
or \emph{directed} (i.e., every edge or link has an intrinsic direction, so that an edge can be directed from $i$ to $j$). The interpretation is that an undirected edge between $i$ and $j$ gives a two-way benefit of flows, i.e., the utility generated by this edge benefits both $i$ and $j$. On the contrary, a directed edge from $i$ to $j$ has a one-way flow of benefits since the utility generated by this edge benefits $i$ but may or may not benefit $j$~\cite{VB-SG:00,AG-SG-JK:06}. Examples of undirected relationships are friendships, 
collaborations (e.g., co-authorship), mutual insurance, etc. Examples of directed relationships arise naturally in communication networks (e.g., someone send an email to another person), loan relationships, subordinate-authority relationships, buyer-seller relationships, etc. We refer to~\cite[Chapter 5]{RPG:10} for more discussion on the use of undirected vs directed graphs in game theoretical modeling.

We stress that not all network formation games in the literature have been studied from a dynamic analysis (indeed, all the classifications presented above only have to do with the static analysis aspect). A first classification from a dynamic analysis is to consider whether the agents, at any time-step of the (discrete) evolution of the game, are~\emph{myopic} (e.g.,~\cite{VB-SG:00,MJ-AW:02,EA-RJ-SM:09}), i.e., they want to only maximize their current utility at any time step, or, otherwise, \emph{foresighted} (e.g.,~\cite{BD-SG-DR:05,FD:06,YS-MvsS:15}), i.e., they also want to maximize their (expected cumulative) future utility. Myopic agents are a more reasonable assumption in contexts where agents live in a big and/or complex network, where it is virtually impossible for a single agent to completely predict how her actions would influence over the rest of the whole network and thus how current actions will affect future utility.
%


\subsection{Contributions}

Our first main contribution is to propose, to the best of our knowledge, a first 
study on the emergence of hierarchical structures from the perspective of network formation games. We propose a game which allows us to study and understand how self-interested or utility-maximizing individuals lead to the emergence of networks with hierarchical structures. In particular, our work shows this through the judicious inclusion of well-motivated hierarchical promoting terms in the utility functions of the agents.  
%
%
We assume all agents are homogeneous, i.e., that we have an egalitarian society, and we also assume that 
all individuals can be of two types: non-consensual or consensual. For each type, we perform a static analysis of the emergent network structures that satisfy the 
solution concept proposed for our game and which 
is closely related to other notions in the classic literature of network formation games. Then, we present a dynamic analysis where we analyze the dynamic formation 
of these networks 
as a result of a stochastic process of pairwise interactions among agents playing better-response  dynamics. 

Our second contribution is a complete static analysis where we show that 
network structures with different levels of hierarchy emerge as \emph{equilibrium networks}, i.e., as solutions of the game, for any type of agents. In particular, individuals within the same level or rank of hierarchy in these structures can only have cooperative relationships. 
Moreover, we provide theoretical results that show the relationship between the structure or topology 
of the formed equilibrium network and the different parameters in the utility function of the agents. Then, we provide intuitive interpretations of these theoretical results and their implications in real-world or practical scenarios. 
Moreover, since the formation of hierarchies occurs within homogeneous agents, 
our results suggest that even in egalitarian societies, hierarchies seem to be a natural outcome of socio-economic relationships. Finally, we show that all possible equilibrium networks that can be
formed by non-consensual agents in any game are a strict subset of the ones that are possible to be formed in the same game by consensual agents. 

Our third contribution is to analyze our proposed network formation game when 
it is played dynamically starting from any fixed initial graph. 
In such dynamic setting, agents are stochastically chosen to perform an action and they 
play their better-response dynamics. We show that the networks, being formed dynamically, eventually converge to some equilibrium network in finite time. 
Finally, we show that for any type of agent, 
the topology of the converged 
hierarchical network, besides depending on the parameters of the utility functions, also depends 
on the particular realization of the underlying stochastic process (i.e., the solution path). 
%
%

We remark that the scope of our work, as seen in the previous contributions, does not include the analysis of the efficiency or social welfare of the resulting equilibrium networks in the game (e.g., we do not explore problems such as determining the factors that lead to a more/less efficient equilibrium network), and, as mentioned in the conclusion, we leave this analysis as future work. Instead, we focus on the full characterization and generation of equilibrium networks as hierarchical structures, from both a static and dynamic analysis.

\subsection{Notation and preliminary modeling}
 
Let $N=\until{n}$ be the finite set of agents, and, throughout the paper,
consider $n\geq 3$.  Let $ij:=(i,j)\in N\times{N}$.  Any element of $N$ is
a node of the directed graph (digraph) $G$, and $G$ is thus defined by a
set of ordered pairs between different elements of $N$.\footnote{Strictly
  speaking, a graph is a pair $(N,E)$, where $N$ is the node set $N$ and
  $E$ is the edge set $E$. For simplicity, we adopt a notation
  traditionally used in the network formation literature (e.g.,
  see~\cite{MOJ:05}) and refer to the graph by its edge set; we let the
  context specify the node set.}  An element $ij\in G$ is called an
\emph{edge} from $i$ to $j$.  Let $G+ij:=G\cup\{ij\}$ and $G-ij:=G\setminus
\{ij\}$.

%
%
We define $\Nin(i,G)=\setdef{j\in N}{ji\in G}$ and $\Nout(i,G)=\setdef{j\in N}{ij\in G}$, with the in-degree of $i$ being $\din(i,G)=|\Nin(i,G)|$. 
Given $G$ and $i\in N$, if $\Nin(i,G)=\emptyset$ ($\Nout(i,G)=\emptyset$) then $i$ is a \emph{source} (\emph{sink}). 
If $ji\in G$ and $ij\notin G$, then we say $ij$ is a \emph{single edge}; and, if $ij,ji\in G$ then we say $i$ and $j$ have an \emph{undirected edge}. A graph is \emph{complete} if there is an undirected edge between any pair of nodes.

Whenever there exists a sequence of edges $ik_1,\dots,k_{m-1}j$ that connects $i$ with $j$ with no edge appearing more than once 
and no node appearing more than in two edges, we say that there exists a path of length $m$ between $i$ and $j$. Whenever $j=i$ in the previous path, we say that this path is a \emph{cycle} containing $i$.  
We denote by $C(G)$ the condensation graph of $G$, and by $\Co(i,G)$ the connected component of $i\in N$. Graph $G$ is \emph{acyclic} if it has no cycles, i.e., it has no paths of one node to itself. A~\emph{directed cycle} is any cycle that contains at least one single edge. 

We refer to Section~\ref{app-basgraph} of the Appendix for basic graph theoretical terminology and concepts that will be used throughout the paper.
%
%

\subsubsection*{New notation and definitions}
%

Given a directed acyclic graph $G'$ and a node $i$ of $G'$, the \emph{level} of $i$, denoted by $\ell_i(G')$, is the largest number of arcs or edges from any source node of $G'$ to $i$. Also, for digraph $G$, the level of $i$ denoted by $\ell_i(C(G))$, is the largest number of arcs in $C(G)$ from any source of $C(G)$ to the strongly connected component containing $i$. See Figure~\ref{fig_conc}(a).

We also define 
$\Nins(i,G)=\Nin(i,G)\setminus\setdef{j\in\Nin(i,G)}{ij\in G}$ and $\dins(i,G)=|\Nins(i,G)|$. 

Consider a graph $G$ and any $i\in N$. Given a node $k\in\Nins(i,G)$, we define $P_i(G,k)=\{i\}\cup \setdef{j\in N}{kj\in G\text{ and }ij,ji\in G}$ and 
$P_i(G)=\cap_{k\in\Nins(i,G)}P_i(G,k)$.  
We use the term \emph{undirected} connected component whenever all edges of the connected component are undirected, and we call it \emph{complete} connected component when additionally it is a complete subgraph.
%

%
%

Given an undirected connected component that contains nodes $i$ and $j$, we say the undirected edge $ij,ji\in G$ is \emph{critical} if 
$\Co(i,G)\neq\Co(i,G-ij)$.

Finally, throughout the paper, we use the terms \emph{agent}, \emph{node} and \emph{individual} interchangeably, as well as the terms \emph{network} and \emph{graph}, and the terms \emph{edge} and \emph{link}. 
 
\begin{figure}[ht]
  \centering
  \subfloat[]{\label{f:1a}\includegraphics[width=0.22\linewidth]{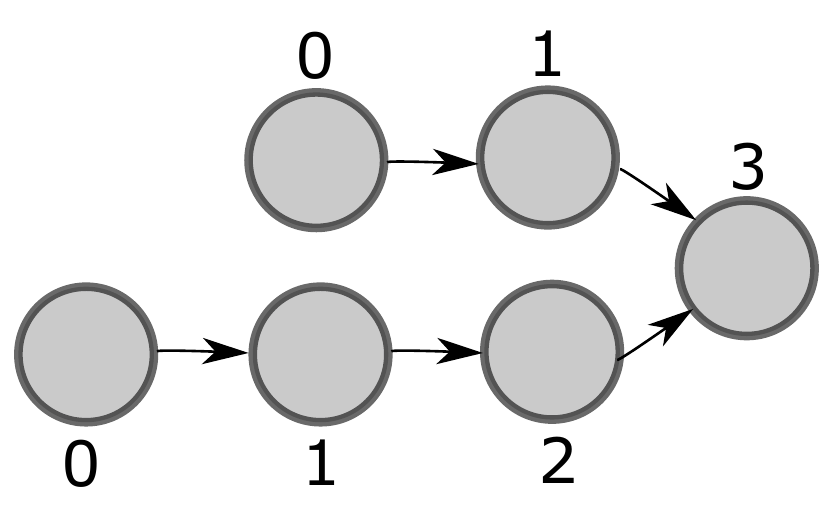}}
    \hfil
  \subfloat[]{\label{f:1b}\includegraphics[width=0.12\linewidth]{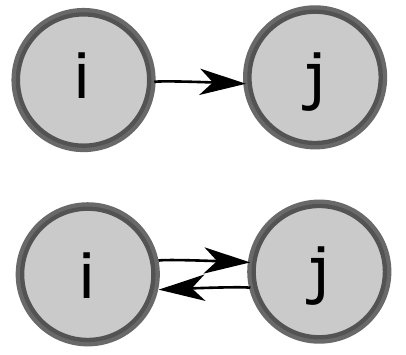}}
  \hfil`
  \subfloat[]{\label{f:1c}\includegraphics[width=0.16\linewidth]{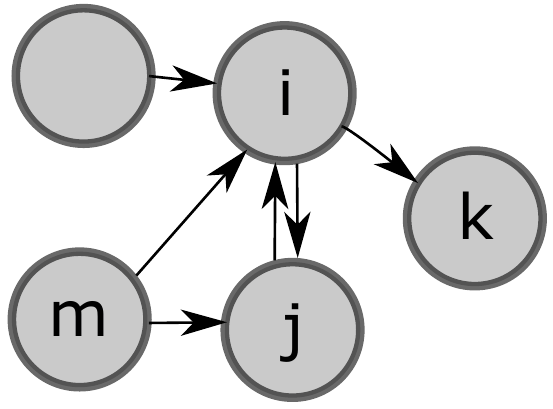}}  
\caption{In (a), we number the levels of all nodes in a directed acyclic graph. In (b), given two agents $i$ and $j$, we represent relationships of subordination in the upper figure (here, $i$ is $j$'s subordinate), and collaboration in the lower figure. In (c), 
$\Nin(i,G)=\{j,m\}$, $\Nin(j,G)=\{i,m\}$, $\Nins(i,G)=\Nins(j,G)=\{m\}$, $\Nin(k,G)=\Nins(k,G)=\{i\}$, 
$P_i(G,m)=P_j(G,m)=P_i(G)=P_j(G)=\{i,j\}$.
}\label{fig_conc}
\end{figure} 
 
\section{Proposed model and game}

In our model, any single edge $ij$ is a pairwise relationship that describes an unambiguous  hierarchical or authority relationship where agent $j$ has a higher hierarchy or exercises some authority over $i$; so we say $j$ is a \emph{supervisor} of $i$, and, respectively, that agent $i$ is a \emph{subordinate} of $j$. 
This simple concept has been used extensively in the game theory literature~\cite{RPG:10}. On the other hand, whenever $i$ and $j$ have an undirected edge, we understand that they have a collaborative or cooperative relationship, and thus, they must share the same hierarchical status and no agent is the subordinate of the other one. Accordingly, we have that 
$\Nins(j,G)$ is the set of $j$'s subordinates, $P_j(G,i)$ is the set of $j$'s collaborators who also have $i$ as a common subordinate, and $P_j(G)$ is the set of $j$'s collaborators who have the same subordinates as $j$. 
See Figure~\ref{fig_conc}(b) and Figure~\ref{fig_conc}(c).

\subsection{Utility definition}
\label{util-def}
We first specify the costs and rewards that any agent incurs according to our proposed model.
\begin{enumerate}
\item \emph{Subordinate-collaborative reward:} From an organizational perspective, this represents the positive reward an individual obtains by being the subordinate or collaborator of another competent agent. For example, consider agent $i$ wants to pursue her own goals (e.g., as an employee), then she will seek direction and help from perhaps an equally or better skilled person (e.g., a competent boss or co-worker).  
This reward can also be interpreted as an exogenous reward (e.g., a salary) that an agent receives by working for or with other people inside an organization. We represent the subordinate-collaborative reward per link 
by the positive constant $\gamma$.
%
%
\item \emph{Hierarchical reward:} Individuals who have a higher hierarchical status or position 
receive an inherent reward for being in a higher position of power. For example, from an organizational perspective, individuals with higher management or directive positions have skill sets that are highly specialized to maintain such high hierarchical positions, and, as a consequence, they are (mostly) better paid in the job market. Moreover, in political organizations and public institutions, individuals with higher hierarchical positions can exert more influence than those in lower positions (e.g., leaders or directives have more influence over ordinary individuals in the society or over their political base). Then, the reward for these individuals 
comes from the strength of their influence and the reach of the consequences of their decisions. 
Mathematically, it is natural that if $i$ and $j$ both appear in some cycle, then agent $i$ has no authority or any hierarchical power over $j$, 
and so $i$ must have the same hierarchical reward as $j$. Therefore, two agents in the same strongly connected component have the same reward. In summary, 
%
%
we model the hierarchical reward for agent $i$ in the graph $G$ as a (strictly) monotonic function which takes as an argument the level of $i$. 
%
%
Formally, the 
%
%
hierarchical reward for an agent $i$ situated in the network $G$ is
%
%
  \begin{equation}
  \label{h1}
   H(\ell_i(C(G)) 
  \end{equation}

We call $H$ the \emph{reward function}, and we assume it is non-negative and (strictly) monotonic (e.g., we have that $H(0)\geq 0$).  Intuitively, the monotonic property of $H$ means that individuals with higher positions of status are better rewarded. 
\item \emph{Management cost:} An individual $i$ who is in charge of a group of people incurs on a management cost directly proportional to the number of her subordinates. There are two possible reasons for this. First, leading or directing more subordinates requires more managerial or logistic work, thus demanding more energy than doing it with a smaller group, since $i$ has to respond to the demands of more people. Naturally, people who collaborate with $i$ (i.e., any $j$ such that $ij,ji\in G$) do not incur a management cost to $i$. Second, from an organizational perspective, a supervisor $i$ is responsible for the group she leads, and she also needs to report to her immediate superiors (i.e., the people, if any, to whom she is a subordinate) about how the people she is responsible for are performing. This reporting task, obviously, requires more work when $i$ is directing larger groups. Observe that management costs must be reduced if $i$ is part of a team of people who also manage or direct her same subordinates since this team will collaborate with her and thus reduce the workload on $i$. 
Formally, we model the management cost as 
 $c\sum_{k\in\Nins(i,G)}\frac{1}{|P_i(G,k)|}$ 
(observe that this quantity is well-defined) for some positive constant $c>0$. The chosen functional form for the management cost captures the qualitative behaviors mentioned above: the larger the number of subordinates of $i$ (i.e., the larger the set $\Nins(i,G)$), the larger the cost; and the larger the number of collaborators (i.e., the larger the set $P_i(G,k)$ with $k\in\Nins(i,G)$), the smaller the cost.
\end{enumerate}

We conclude that the utility of an individual is the sum of her subordinate-collaborative and hierarchical rewards minus her management cost. 
\begin{definition}[Utility function]
\label{def_u}
The utility that an individual $i$ incurs given a graph $G$ is
\begin{equation}
\label{util_eq}
u_i(G)=|\Nout(i,G)|\gamma+H(\ell_i(C(G))  -c\sum_{k\in\Nins(i,G)}\frac{1}{|P_i(G,k)|}
\end{equation}
\end{definition}
Clearly, if $i$ has no connections and no other agent is her subordinate in $G$, then $u_i(G)=0$. 

\textbf{\emph{Example 1: }}We present a couple of graphs for which we compute the utilities of some agents to better  illustrate the calculation using equation~\eqref{util_eq}. Consider Figure~\ref{fig_example}. If $G$ is as in (a), then
\begin{equation*}
u_i(G)=2\gamma+H(1)-\frac{3}{2}c, \quad u_j(G)=\gamma+H(1)-\frac{1}{2}c, \quad u_k(G)=H(2)-c, \quad u_m(G)=2\gamma+H(0).
\end{equation*}
If $G$ is as in (b), then
\begin{equation*}
u_i(G)=2\gamma+H(2)-c, \quad u_j(G)=\gamma+H(2)-2c, \quad u_k(G)=\gamma+H(2)-c.
\end{equation*}
%
\begin{figure}[ht]
  \centering
  \subfloat[]{\label{f:11-b}\includegraphics[width=0.19\linewidth]{figure_example_1.pdf}}
    \hfil
  \subfloat[]{\label{f:14-b}\includegraphics[width=0.24\linewidth]{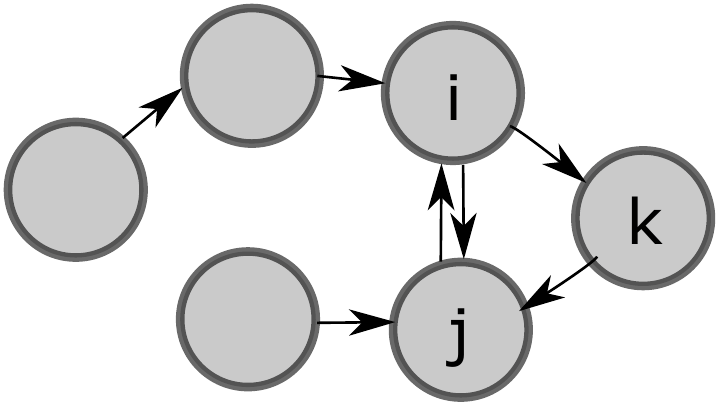}}
\caption{Network examples.
}\label{fig_example}
\end{figure}

\subsection{Game formulation and solution concept}

Having defined the utility that individuals obtain from their membership in the network $G$, we now complete the definition of our proposed network formation game in strategic (or normal) form~\cite{DF-JT:91}. 
%
%
An agent $i\in N$ can only perform two types of actions over the network $G$: establishing an edge $ij$ so that $ij\in G$, or severing an existing edge $ij$ so that $ij\notin G$. 
Regarding the agent's type, we say $i$ is \emph{consensual} whenever $i$ requires $j$'s agreement to establish the link $ij$; otherwise, $i$ is \emph{non-consensual} whenever she can perform such action irrespective of $j$'s agreement. 

\begin{definition}[Hierarchical network formation game]
A \emph{hierarchical network formation game} $\mathcal{G}$ is defined by the set of homogeneous agents $N=\until{n}$, with $n\geq 3$, the utility function for any agent $i\in N$ as defined in~\eqref{util_eq}.  
Let $\mathcal{A}_i=\{0,1\}^n$ 
be the action space of $i\in N$, so that any $a_i=(a_{ij})\in \mathcal{A}_i$ is an action (or pure strategy) for agent $i$. 
Given a graph $G$, we say agent $i$ has taken action $a_{ij}=1$ if $ij\in G$; and, otherwise, has taken action $a_{ij}=0$ if $ij\notin G$.
Let $(\gamma,H,c)$ be the \emph{utility parameters} that define the payoff of any agent as in Definition~\ref{def_u}, and $a=(a_1,\dots,a_n)$ be the action (or pure strategy) profile which takes values from the set $\mathcal{A}=(\mathcal{A}_1,\dots,\mathcal{A}_n)$. 
Then, we represent the hierarchical network formation game $\mathcal{G}$ by the tuple $(N,u,\mathcal{A},\gamma,H,c)$ along with the specification of whether all agents are consensual or non-consensual.
\end{definition}
%
%
%
As it is usual in the literature on network formation games, we only consider pure strategies throughout the paper. We also consider that any agent $i\in N$ is self-interested or utility-maximizing and so she wants to chose the action profile that maximizes her utility~\eqref{util_eq}.

\begin{remark}[About the game's modeling scope]
\label{interpret}
Firstly, we observe that an agent can only
add or remove an outgoing edge, i.e., she decides who supervises her or who she wants to work for. 
This captures scenarios where people can decide to join (or apply for) a specific team at work or some program, and are able to quit when desired. In the case of non-consensual agents, since an agent always accepts an incoming edge, this assumption does not capture scenarios where managers have more power selecting their employees compared to the power employees have in selecting their managers. However, this scenario is captured in the case of non-consensual agents: agents have the managerial power to select or hire which individuals will be their subordinates from the pool of applicants.

Secondly, we observe that an agent cannot remove an incoming edge. 
This models a sense of responsibility on the agents: 
an agent accepts to supervise an individual 
(whether the type is non-consensual or consensual) knowing that she will not be able to stop supervising 
if she ``later regrets" having such responsibility. This also models situations where an \emph{ex-ante} contract protects subordinates from being indiscriminately fired, but does not model scenarios where there is more managerial power in firing employees. 
\end{remark}

%
%
%
%
%
%

We propose and use the following solution concept for hierarchical network formation games:

\begin{definition}[Equilibrium network]
\label{eq_net}
A network $G$ is an \emph{equilibrium network} when, for any $i,j\in N$,
\begin{enumerate}
\item if $ij\in G$ then $u_i(G)>u_i(G-ij)$ and;
\item if $ij\notin G$, then 
\begin{itemize}
\item for non-consensual agents: $u_i(G)>u_i(G+ij)$,
\item for consensual agents: either $u_i(G)>u_i(G+ij)$, or, if $u_i(G)\leq u_i(G+ij)$ then $u_j(G)>u_j(G+ij)$.
\end{itemize}
\end{enumerate}
\end{definition}

\begin{remark}[Connection to other solution concepts]
The concept of an 
equilibrium network becomes the concept of strict Nash equilibrium whenever the agents are non-consensual, and it becomes the concept of pairwise stability when the agents are consensual (with the difference that our concept employ strict inequalities in contrast to the definitions of pairwise stability found in other works, e.g.,~\cite{MOJ-AW:96}). 
\end{remark}


If there exists some edge $ij$ with $i,j\in N$ that can be added (or severed) so that $u_i(G+ij)\geq u_i(G)$ (or $u_i(G-ij)\geq u_i(G)$), i.e., so that some condition on the utilities of $i$ as described in Definition~\ref{eq_net} is violated, then we say that $i$ has an \emph{incentive} or \emph{intention} to establish (or sever) the link $ij$ respectively. Obviously, 
when any agent has an incentive to severe some link, then $G$ is not an equilibrium network for any type of agents. However, if agent $i$ has an incentive to establish a link, then 
$G$ is not an equilibrium network for non-consensual agents, but may or may not be one for consensual agents (recall that for consensual agents, the addition of the edge $ij$ also depends on $j$). 

We conclude by providing 
some further interpretations and intuition on the agents' types. A non-consensual agent $i$ can be thought of being ``greedy" because $i$ accepts any agent that intends to be her subordinate or collaborator. Intuitively, we can think of greediness from two perspectives: as a \emph{bounded rationality} aspect, since $i$ can accept more subordinates even though this would mean a negative impact on her utility; or perhaps as a \emph{foresightedness} feature of $i$, since $i$ may accept any subordinate with the hope of having a larger future hierarchical reward. On the other hand, we can think of a consensual agent $i$ as being ``non-greedy" since she does not want to acquire more subordinates than what she can really handle, i.e., if acquiring more subordinates will affect her utility negatively. 

\section{Static Analysis}

\subsection{Hierarchical structures}
\label{sh_sec}

We introduce a useful class of networks that abstracts and represents the concept of hierarchies.  
%
%
Intuitively, agents with a higher level have higher positions of power, and formally, this is 
true, since agents with a higher level have a greater hierarchical reward (see equation~\eqref{h1}). 
When $k$ different connected components have the same level $\ell$, we say there are $k$ components per level $\ell$. If the maximum level that any agent in the network has is $\ell$, then we say the network has $\ell+1$ levels.

\begin{definition}[Hierarchical structures]
\label{def1}
A (weakly) connected digraph $G$ 
is a \emph{hierarchical structure} if each connected component is an undirected subgraph, and it is additionally a \emph{sequential hierarchy} if there is only one connected component per any level, and, if there is more than one level, then any node has a single edge towards any other node of a higher level.
%
\end{definition}
%
%
%
%
A hierarchical structure $G$ has the particularity that any of its connected components are composed of agents that form strong collaborative units or \emph{teams}, i.e., all agents in a connected component are collaborators. 
Thus, $G$ is naturally a hierarchical representation of the organization, power and/or influences among 
teams of people (with the understanding that a team may be composed of a single individual). Now, a sequential hierarchy $G$ is a hierarchical representation in which there is only one team at each rank or level of the hierarchy. Observe that a particular trivial case of a hierarchical structure which is also a sequential hierarchy is when the network is complete, i.e., all agents form a unique team. See Figure~\ref{fig_example_hs} for examples of hierarchical structures.

%
%
%
%
%

Observe that a \emph{pyramidal structure} is formed by a hierarchical
structure with two or more levels, generally with one connected component
per any level, and such that the number of agents that have a lower level
is greater than the ones that have a higher one. This type or class of
graphs is easily found in real life hierarchies; for example, military
organizations (there are less generals than colonels, but less colonels
than lieutenants, etc.) and clerical hierarchies in the Catholic
church~\cite{Cath} (there is only one Pope, and more bishops, but there are
less bishops than priests) have this type of hierarchies. Moreover, it is
known that many companies naturally adopt a pyramidal structure for its
management organizations~\cite{WP-RF-GA:07}, with the CEO or president at
the top, followed by a small executive leadership or vice-presidents,
followed by tiers of middle managers, and all the way down to the lowest
level employees.
%
%
\begin{figure}[ht]
  \centering
  \subfloat[]{\label{f:21}\includegraphics[width=0.4\linewidth]{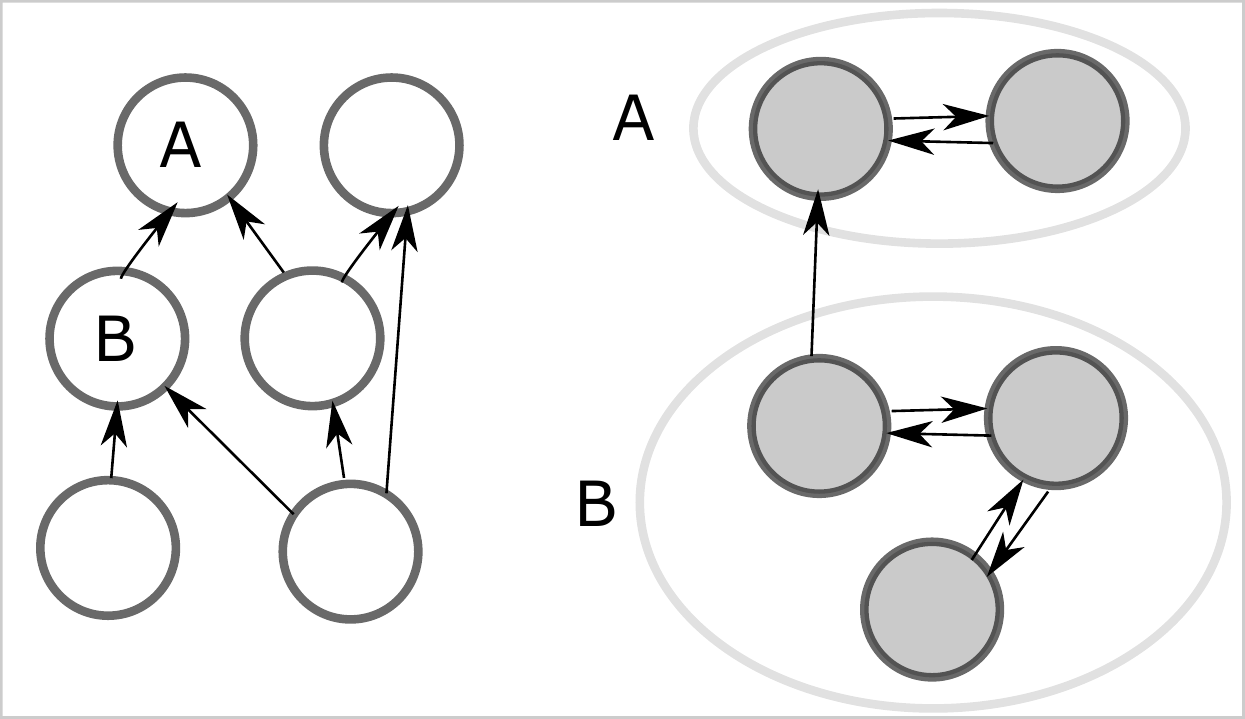}}
    \hfil
  \subfloat[]{\label{f:22}\includegraphics[width=0.4\linewidth]{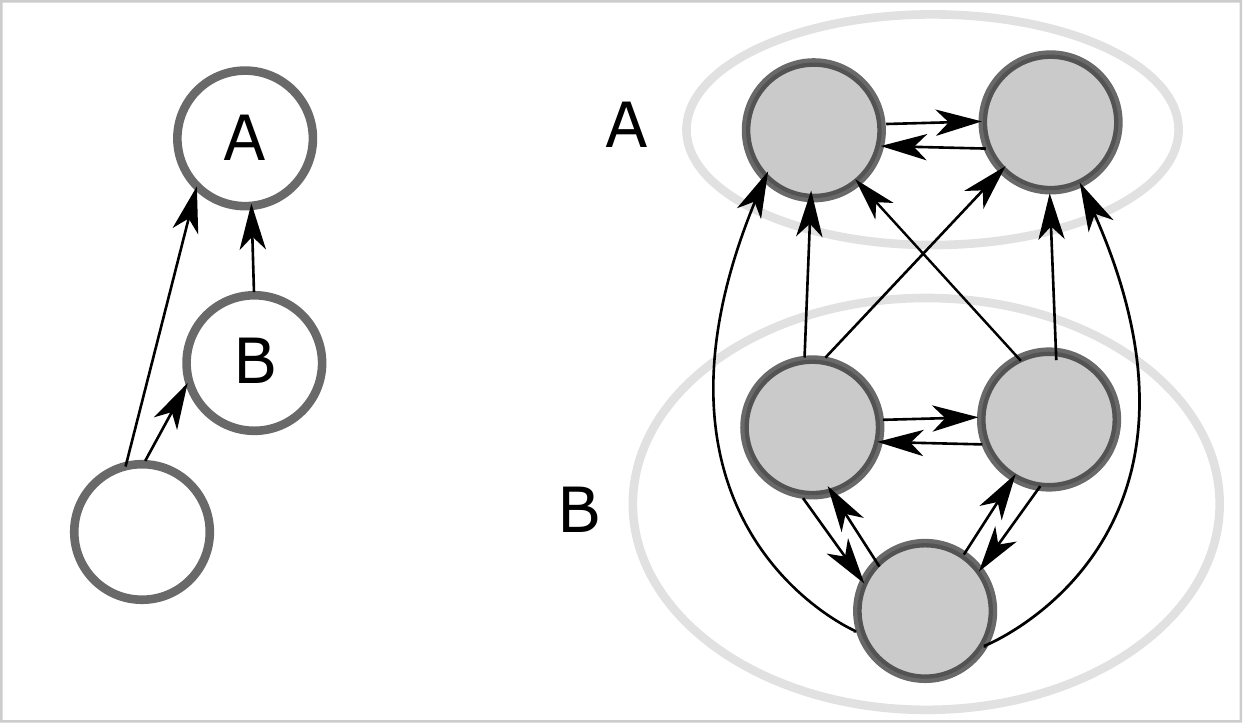}}
\caption{Consider each subfigure, (a) and (b), as an example of a digraph $G$ which has its condensation graph $C(G)$ represented by the graph with white-colored nodes on the left, 
and which has a subgraph 
corresponding to the nodes $A$ and $B$ of $C(G)$ represented by a graph with gray-colored nodes on the right. 
In (a) we have a particular example of a hierarchical structure which is not a sequential hierarchy, whereas in (b) we have a hierarchical structure which is a sequential hierarchy. In both cases, all agents corresponding to node $A$ in $C(G)$ have level $2$, whereas all agents corresponding to node $B$ in $C(G)$ have level $1$.}
\label{fig_example_hs}
\end{figure}

\subsection{Solution Analysis}

Now, we introduce the main results for this section. All relevant proofs can be found in Section~\ref{proofs-results} of the Appendix.

\begin{theorem}[Non-consensual agents]
\label{th:non-c}
Consider non-consensual agents. A network $G$ is an equilibrium network for some game $\mathcal{G}$ if and only if it is a sequential hierarchy that satisfies the following conditions:
\begin{enumerate}
\item \label{cond-1}If the network has $k>1$ levels, then the following condition holds:
\begin{equation}
\label{cond_im}
H(\ell)-H(\ell-1)>\gamma+c\left(\frac{|P_i(G)|+|\setdef{j}{\ell_j(C(G))<\ell}|}{|P_i(G)|(|P_i(G)|+1)}\right)
\end{equation}
for any $i,j\in N$ such that $\ell_i(C(G))=\ell$ and $\ell_j(C(G))=\ell-1$ with $0<\ell\leq k-1$.
\item \label{cond-2}If there exists some critical edge $ij,ji\in G$, then 
\begin{equation}
\label{cond_2}
H(\ell+1)-H(\ell)<\gamma+c\left(1+\frac{|\setdef{j}{\ell_j(C(G))<\ell}|}{2}\right) \,\text{ where }\,\ell_i(C(G))=\ell\geq 0.
\end{equation}
\end{enumerate}
Moreover, there exist sequential hierarchies that are not equilibrium networks for any hierarchical network formation game.
\end{theorem}

\begin{remark}[Interpretation of Theorem~\ref{th:non-c}]
\label{rem:inter-nonc}
The importance of equations~\eqref{cond_im} and~\eqref{cond_2} is that they relate the agents' utility parameters $(\gamma,H,c)$ with the specific topology of the sequential hierarchy $G$ resulting from the hierarchical network formation game with non-consensual agents.

We start by interpreting equation~\eqref{cond_im}. A first observation is
that this equation specifies how the increment in the hierarchical reward
between the current level $\ell$ and the previous level $\ell-1$, i.e.,
$\Delta(H):=H(\ell)-H(\ell-1)$, restricts the number of agents that can
exist at level $\ell$, i.e., $|P_i(G)|$. Interestingly, the larger the
increment $\Delta(H)$, the smaller the number of agents that can possibly
exist with level $\ell$ in an equilibrium network.\footnote{To analyze this
  fact in the right-hand side of equation~\eqref{cond_im} observe that the
  function $x\mapsto\frac{x+a}{x(x+1)}$ with $x>0$ and positive constant
  $a$, is monotonically decreasing.} From an organizational point of view,
this can be interpreted as follows: whenever there is a large difference in
power between two consecutive levels in the hierarchy of an organization,
we can then expect to find less people in the higher level. For example, in
a company, there are fewer executives than managers. Therefore, the
equilibrium networks from Theorem~\ref{th:non-c} are able to model a
pyramidal structure with $N$ levels by, for example, enforcing monotonic
increments in the hierarchical reward (e.g.,
$H(1)-H(0)<H(2)-H(1)<\dots<H(N)-H(N-1)$). Moreover, this reflects a natural
organizational phenomenon: the salary difference between higher positions
in an organization are more accentuated than between lower positions.

A second observation of equation~\eqref{cond_im} is the following: if
inequality~\eqref{cond_im} is satisfied for greater values of the
management cost parameter $c$, then it is possible to find a greater number
of agents with level $\ell$.\footnote{Please, see the previous
  footnote.} This has an intuitive explanation: more agents can be found on
level $\ell$ because this implies that any agent in this level has more
collaborators with whom to reduce the management cost that otherwise would
be larger due to the larger value of $c$.
%
%

We now interpret equation~\eqref{cond_2}. First observe that, since the equilibrium networks are sequential hierarchies, all connected components are complete subgraphs. Therefore, whenever there is a critical edge, it must belong to a connected component of two nodes. Then, equation~\eqref{cond_2} says that agent $i$ with level $\ell$ and who has a critical edge, has no incentive in severing its collaborative tie with $j$ if, for example, her subordinate-collaborative reward $\gamma$ is large enough to counteract the increase on the hierarchical reward that she would get by severing her critical edge. This condition seems restrictive because it implies an upper bound on $\Delta(H)$, however, we remark that this condition is only present whenever a critical edge exists, i.e., whenever a group of only two people exist for some level. Indeed, later in this paper, in Theorem~\ref{th_conv}, a result is presented where agents dynamically play the hierarchical network formation game and a network with no critical edges is formed.

\end{remark}

\begin{theorem}[Consensual agents]
\label{th:c}
Consider consensual agents. A network $G$ is an equilibrium network for some game $\mathcal{G}$ if and only if $G$ is composed by one or multiple hierarchical structures and $G$ satisfy the following conditions:
\begin{enumerate}
%
\item\label{lab_1} If a hierarchical structure in $G$ has $k>1$ levels, then 
\begin{equation}
\label{cond_im1}
H(\ell)-H(\ell-m)>\gamma+c\sum_{p\in\Nins(i,G)\setminus\{j\}}\left(\frac{|P_i(G+ij,p)|-|P_i(G,p)|}{|P_i(G,p)||P_i(G+ij,p)|}\right)+\frac{c}{|P_i(G,j)|}> 0
\end{equation}
for any $i,j\in N$ belonging to this hierarchical structure and such that $ji\in G$, with $\ell_i(C(G))=\ell\leq k-1$ and $\ell_j(C(G))=\ell-m$ for any appropriate $0\leq m\leq\ell$.
\item \label{lab_2}If there exists $i,j\in N$ such that $\ell_i(C(G))=\ell\geq 0$ and $\ell_j(C(G))=\ell-m\geq 0$ for any appropriate $0\leq m\leq\ell$, and such that there exists no path from $i$ to $j$ (whether $i$ and $j$ belong to the same hierarchical structure or not), then 
\begin{equation}
\label{cond_im2}
H(\ell+1)-H(\ell-m)<c.
\end{equation} 
%
%
\item \label{cond-2a}If there exists a critical edge $ij,ji\in G$ such that $\ell_i(C(G))=\ell\geq 0$ and $\ell_i(C(G-ij))=\ell+1$, then 
\begin{equation}
\label{cond_2a}
H(\ell+1)-H(\ell)<\gamma+c\sum_{p\in\Nins(i,G)}\left(\frac{|P_i(G,p)|-|P_i(G-ij,p)|}{|P_i(G,p)||P_i(G-ij,p)|}\right)+c.
\end{equation}
\end{enumerate}
%
%
Moreover, there exist hierarchical structures, which are possibly sequential hierarchies, 
that are not equilibrium networks for any hierarchical network formation game.
%
%
\end{theorem}

\begin{remark}[Interpretation of Theorem~\ref{th:c}]
From Theorem~\ref{th:c} we conclude that consensual agents allow the formation of a wider variety of hierarchical structures than non-consensual agents. Indeed, given a game $\mathcal{G}$ with non-consensual agents, any of its equilibrium networks is also an equilibrium network for the same game but with consensual agents. We also remark that consensual agents allow equilibrium networks that are not connected. See Figure~\ref{fig_ex_emerged}. 

Therefore, we note that equations~\eqref{cond_im1} and~\eqref{cond_2a} from Theorem~\ref{th:c} are the more general versions of equations~\eqref{cond_im} and~\eqref{cond_2} from Theorem~\ref{th:non-c},  respectively, and, despite having a more intricate expression, inherit a similar qualitative interpretation from their counterparts (as explained in Remark~\ref{rem:inter-nonc}).  

Equation~\eqref{cond_im2} can only be present in the case of consensual agents (i.e., it has no counterpart for non-consensual agents) for two reasons. The first reason is that it must be satisfied, for example, when the equilibrium network consists of two or more hierarchical structures, or if the equilibrium is a hierarchical structure whose condensation graph has multiple sources. The second reason is that this condition is a result of the non-consensual interaction: considering the statement~\ref{lab_2} of the theorem, if agent $i$ wants to establish a connection with $j$ (observe that $i$ has a higher level than $j$), $j$ will not consent $i$'s intention of becoming her subordinate if the increase of management cost is larger than the increase of hierarchical reward that $j$ would obtain from having $i$ as a subordinate.
\end{remark}

It is easy to construct games with consensual agents in which one or multiple sequential hierarchies (including complete networks) are equilibrium networks. The next proposition gives other non-trivial examples of equilibrium networks for a game with consensual agents.

\begin{proposition}[Examples of equilibrium networks for consensual agents]
\label{lem-ex}
Consider any hierarchical structure $G$ with more than one level, one complete connected component per any  level whose size is greater than two, 
and such that the only existing edges in the graph are as follows: for any $i,j\in N$ with  $\ell_i(C(G))=\ell$ and $\ell_j(C(G))=\ell+1$, 
$pk\in G$ for any $p\in\Co(i,G)$, $k\in\Co(j,G)$. Then, 
\begin{enumerate}
\item \label{s11}$G$, or 
\item \label{s22}any network composed by two or more (isolated) hierarchical structures with the same properties as $G$ and the additional property that $|P_i(G)|\geq |P_j(G)|$ for any $i,j\in N$ with $\ell_j(C(G))=\ell_i(C(G))-1>0$. 
%
\end{enumerate}
is an an equilibrium network for some game with consensual agents.
\end{proposition}

\begin{figure}[t]
  \centering  
  \label{f:3}\includegraphics[width=0.37\linewidth]{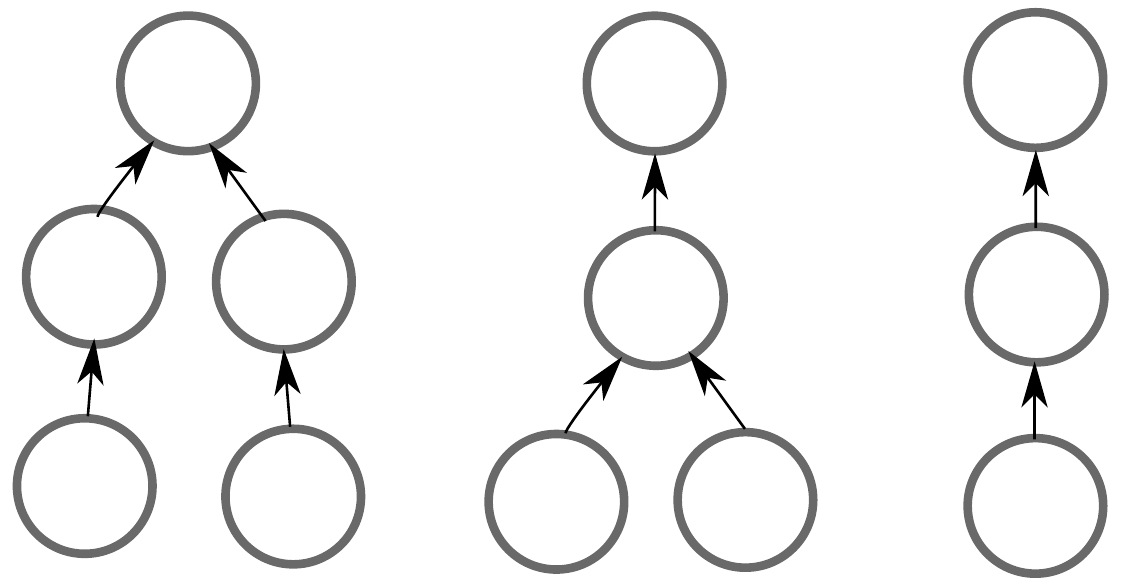}
\caption{We depict the condensation graph of a network $G$ containing three isolated hierarchical structures that are possible to result as a solution to a hierarchical network formation game with consensual agents (provided the adequate utility parameters exist).  
Notice that each isolated network has the topology of another possible
equilibrium network for consensual agents, but not for non-consensual
ones.}
\label{fig_ex_emerged}
\end{figure}
\begin{remark}[Existence of hierarchical structures that are not solutions of the game]
As mentioned by both Theorems~\ref{th:non-c} and~\ref{th:c}, there exist hierarchical structures that cannot be equilibrium networks of any hierarchical network formation game, i.e., there exist hierarchical structures 
which are distinct from the equilibrium networks of all possible tuples 
$(N,u,\mathcal{A},\gamma,H,c)$
that define a game with either consensual or non-consensual agents.
%
\end{remark}

\section{Dynamic Analysis}

We now explore the emergence of hierarchical structures when agents play a hierarchical network formation game dynamically, which we call the \emph{network formation process}. 

\begin{definition}[Network formation process]
\label{d_proc}
Consider a hierarchical network formation game $\mathcal{G}$ and any fixed graph $G_o$. At each time step $t=0,1,\dots$:
\begin{enumerate}
\item \label{s1}Consider the action profile at time $t$ as $a(t)=(a_1(t),\dots,a_n(t))$ and the respective emergent network as $G(t)$. If $t=0$, then consider $G(0)=G_o$.
\item \label{s2}Select a pair of agents $(i,j)$, $i\neq j$, randomly and independently from $N\times{N}$ 
with all possible pairs of different agents only having positive probabilities of being selected according to some time-invariant distribution.
\item \label{s3}Then:
\begin{enumerate}
\item For non-consensual agents: 
  \begin{equation}
  \label{act_br}
    a_{ij}(t+1)\begin{cases}
    =\arg\max_{a^*_{ij}\in\{0,1\}}u_i(a_{ij}^*,a_{-ij}(t)), \qquad &\text{ if }\max_{a^*_{ij}\in\{0,1\}}u_i(a_{ij}^*,a_{-ij}(t))>u_i(a(t))\\
    \in \{0,1\}\setminus\{a_{ij}(t)\}, \qquad &\text{ otherwise}
    \end{cases}.
  \end{equation}
%
%
\item For consensual agents: let first $i$ consider the action $a^*_{ij}$ as if she was a non-consensual agent. If $ij\in G(t)$, then $a_{ij}(t+1)=a^*_{ij}$. Now, if $ij\notin G(t)$ and $a^*_{ij}=1$, then since $i$ is a consensual agent, we have that $a_{ij}(t+1)=1$ if and only if $u_j(G(t)+ij)\geq u_j(G(t))$ (otherwise, obviously, $a_{ij}(t+1)=0$). 
%
\end{enumerate}
%
\end{enumerate}
\end{definition}

\begin{remark}
The network formation process, as clearly seen in Definition~\ref{d_proc}, implements a \emph{better-response dynamics}\footnote{We remark that the network formation process implements a better-response dynamics because an agent $i$ can only consider the addition or deletion of a single edge per time step. If, on the other hand, $i$ is able to select the best pure strategy in her action space $\mathcal{A}_i$ per time step, i.e., being able to alter multiple edges at the same time, we would say the network formation process implements a \emph{best-response dynamics}.} for each agent, in which each utility-maximizing agent plays the game sequentially and myopically, i.e., they maximize their current utility (given the constraint of their actions) without any foresightedness or being forward-looking. Note, however, that when the agent's utility is the same irrespective of what action she can take in the next time step, she will prefer changing her action in the next time step (since the action space has only two states, this simply means switching her current action in the next time step). 
Myopic agents are a natural assumption when networks are big or complex, since in this case it is hard for an agent to trace all the past actions from every other member of the network or predict what the repercussion of her actions will be over the rest of agents. Likewise, myopic agents have the least cognitive and informational requirements about the history of the played game: they only need to remember what happened in the previous time-step. These assumptions have been widely used in the literature, e.g.,~\cite{VB-SG:00}. 
We also remark that pairwise selection of agents in dynamic settings of network formation games (as in step~\ref{s2} from Definition~\ref{d_proc}) has been widely used in the literature of network formation games, e.g.,~\cite{MJ-AW:02}.
\end{remark}

Formally speaking, the network formation process induces a homogeneous Markov chain $\mathcal{M}$ whose state space is the set of all possible graphs with $N$ nodes (defined by all possible action profiles the agents can take), and whose transition probabilities are immediately defined as an induced measure from the underlying pairwise selection process defined in step~\ref{s2} from Definition~\ref{d_proc}. See Figure~\ref{fig:n_t} for an example of a one-step transition in the network formation process.

We proceed to analyze the convergence properties of the network formation process. All relevant proofs  can be found in Section~\ref{proofs-results1} of the Appendix.
\begin{figure}[ht]
  \centering  
  \label{f:3}\includegraphics[width=0.58\linewidth]{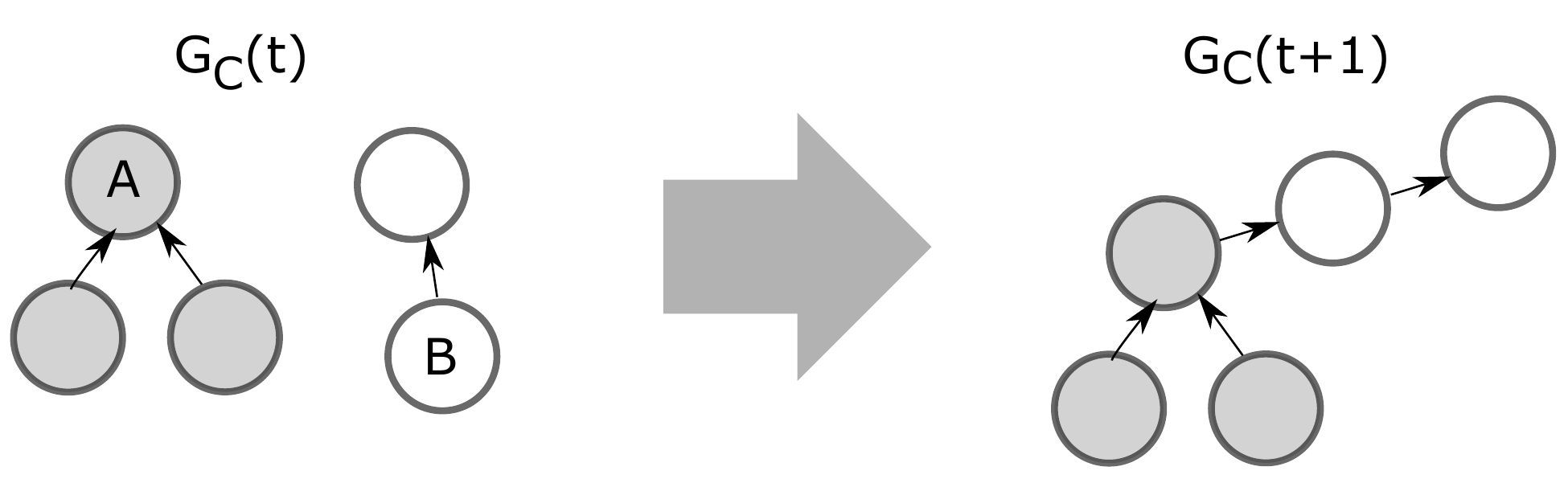}
\caption{Example of a one-step transition in the network formation process. At some time $t\geq 0$, we have an emergent network $G(t)$ with its condensation graph $G_C(t)$. Assume the pair of agents $(i,j)$ is chosen at time $t$ by the network formation process, with $i$ such that $\Co(i,G(t))$ corresponds to node $A$ and $j$ such that $\Co(j,G(t))$ corresponds to node $B$ in $G_C(t)$. Then $G_C(t+1)$, which is the condensation graph of $G(t+1)$, is as shown in the figure for non-consensual agents, and for consensual agents if and only if $H(2)-H(0)\geq c$.}
\label{fig:n_t}
\end{figure}

\begin{definition}[Invariant networks]
\label{def_inv_net}
Consider the network formation process. An \emph{invariant network} $G^*$ is any network such that if $G(t)=G^*$ then, with probability one, $G(t+1)=G^*$. 
\end{definition}
%
%
\begin{proposition}
\label{prop1}
A network is invariant for the network formation process if and only if it is an equilibrium network.
\end{proposition}
\begin{proof}
Consider any invariant network $G^*$ and any $(i,j)\in N\times{N}$ (with $i\neq j$). Consider first non-consensual agents. If $ij\in G^*$ then it must follow that $u_{i}(G^*)>u_{i}(G^*-ij)$; otherwise, if $u_{i}(G)\leq u_{i}(G^*-ij)$, then item~\ref{s3} from Definition~\ref{d_proc} implies that $i$ will change her action in the next time step and thus the topology of $G^*$, which gives a contradiction. A similar proof can be done to show that if $ij\notin G^*$, then $u_{i}(G^*)>u_{i}(G^*+ij)$; and by Definition~\ref{eq_net}, conclude that $G^*$ is an equilibrium network. Now, for consensual agents, it only suffices to analyze what happens when $ij\notin G^*$ and $u_i(G^*)\leq u_i(G^*+ij)$, i.e., the case in which, agent $i$'s action in the next time step would be $1$ if she was non-consensual. 
In such case, according to item~\ref{s3} from Definition~\ref{d_proc}, it must follow that $u_j(G^*+ij)<u_j(G^*)$, so that $i$ has no consent from $j$ and will not change her current action $0$ in the next time step. This implies, according to Definition~\ref{eq_net}, that $G^*$ is an equilibrium network.
%

Given an equilibrium network, it is easy to show it is also an invariant network from Definition~\ref{eq_net} and the fact that only a pair of agents is selected per time step in the network formation process.
\end{proof}
Formally speaking, Proposition~\ref{prop1} tells us that a network $G$ is an absorbing state of the Markov chain $\mathcal{M}$ if and only if it is an equilibrium network. 

\begin{problem}[Important problems]
The first and most important problem to solve is whether (i) the network formation process, which results from the agents playing (myopic) better-response dynamics starting from a fixed arbitrary initial graph, 
can lead convergence 
to an equilibrium network, or (ii) if the individuals may never settle down in forming an invariant network so that their formed networks are continually oscillating between different configurations. Once convergence to an equilibrium network is proved, 
a second
natural problem is the further exploration of the specific structures or topology of the formed equilibrium networks.
%
\end{problem}

\begin{theorem}[Convergence of the network formation process]
\label{th_conv}
Consider the network formation process, with an arbitrary initial network $G_o$, for a game $\mathcal{G}$ with a fixed type of agents and such that $H(\ell+1)-H(\ell)> \gamma+cN$ for any integer $0\leq \ell \leq N-1$. With probability one, $G(t)$ converges in finite time to an equilibrium network with no critical edges and such that it is connected (according to the conditions stipulated in Theorem~\ref{th:non-c} or Theorem~\ref{th:c} depending on the type of agents).
\end{theorem}

The following corollary follows from the proof construction of Theorem~\ref{th_conv}.

\begin{corollary}[Convergence starting from an empty graph]
\label{col_conv}
Consider the network formation process with an empty initial network $G_o=\emptyset$, for a game $\mathcal{G}$ with a fixed type of agents and such that $H(\ell+1)-H(\ell)> \gamma+cN$ for any integer $0\leq \ell \leq N-1$. With probability one, $G(t)$ converges in finite time to an equilibrium network with only single edges and such that it is connected (according to the conditions stipulated in Theorem~\ref{th:non-c} or Theorem~\ref{th:c} depending on the type of agents).
\end{corollary}

\begin{remark}[Analysis of results]
\label{an-res}
There are important remarks to make about the convergence results for consensual agents. 
In general, the final or converged network from the network formation process (for any type of agents) 
depends on both the initial network and the specific realization path of the underlying stochastic process, i.e., the final network is a random function of the initial conditions. This is known as~\emph{path dependence}~\cite{VB-SG:00}: depending on the initial graph, there may be a positive probability of converging to more than one invariant network, provided this is allowed by the utility parameters and the topology of the initial graph. This gives the intuition that in egalitarian societies, the topology of the hierarchies that may emerge have a random component in it. See Example 2 for an illustration of this phenomenon. 
%
Finally, notice how starting the network formation process from an empty graph greatly reduces the final topology of the possible converged equilibrium network (e.g., compare Corollary~\ref{col_conv} to Theorem~\ref{th_conv}).
\end{remark}

We provide a small example to illustrate the observations made in Remark~\ref{an-res} regarding path dependence.
\\

\textbf{\emph{Example 2: (Simple examples of network formation processes and path dependence) }}Consider the network formation process starting from an empty graph with three agents, i.e., $N=\{1,2,3\}$ for some game.

\noindent
Consider first consensual agents in a game with utility parameters such that $H(1)-H(0)=H(2)-H(1)>\gamma+c$, and consider $(a,b,c)$ to be any permutation of elements of $N$. Then, if the first three selected pairs are $\{(a,b),(a,c),(c,b)\}$, then the process converges to the invariant network $G_1(a,b,c)=\{ab,ac,cb\}$ with probability one. However, if the first two selected pairs are $\{(a,b),(b,c)\}$, then the process converges to the invariant network $G_2(a,b,c)=\{ab,bc\}$ with probability one. We conclude that, when starting from an empty graph, the equilibrium networks $G_1$ and $G_2$ (for different values of $(a,b,c)$) have a positive probability of being the invariant network towards which the network converges to, i.e., there is path dependence.

\noindent
Now, consider non-consensual agents. If the utility parameters of the game are such that $0<H(1)-H(0)<\gamma+c$, then the converged network is a complete network with probability one, i.e., no matter what is the realization of the stochastic process, the network always converges to a complete one. 
Now, assume $H(1)-H(0)>\gamma+c$. First consider $H(2)-H(1)>\gamma+c$. Then, the converged network has a structure of the form $G=\{ab,ac,bc\}$, where any possible permutation $(a,b,c)$ of elements of $N$ has a positive probability of appearing. Similarly, if $H(2)-H(1)<\gamma+c$, then the converged network has a structure of the form $G=\{ab,bc,cb,ac\}$ where any possible permutation $(a,b,c)$ of elements of $N$ has a positive probability of appearing. In conclusion, if $H(1)-H(0)<\gamma+c$ then there is path dependency; otherwise, if $H(1)-H(0)>\gamma+c$ then there is no path dependence.
%
%


\section{Conclusion}

In this paper we propose, to the best of our knowledge, the first network formation game model that explains the emergence of hierarchical structures in a society. We consider both non-consensual and consensual agents, and study the different topologies of hierarchies that can emerge as solutions of a game. Moreover, we develop a network formation process in which the agents play stochastically the game sequentially and myopically, and characterize the possible final formed networks. We learn that the introduction of hierarchical rewards that increasingly benefit the agents of higher hierarchy or power, is crucial in determining both the total number of levels in the final hierarchical structure and the maximum number of individuals that can share the same level. Moreover, we learn that it is plausible that individuals within the same level cooperate with each other.

Our paper motivates future work on the analysis of hierarchical structures, like, for example, the proposal of new models that can display either a narrower or a wider variety of other forms of hierarchies than the ones presented in this paper. One way to possibly achieve this, is by the judicious incorporation of heterogeneous agents in the game that may have different structures of subordinate-collaborative or hierarchical rewards. Another possible future work is to define different dynamic processes of network formation that perhaps can drive the agents to different equilibrium selection outcomes and study why such phenomena may happen. 

As pointed out in Remark~\ref{interpret}, a direction of future research is to extend our proposed game to allow agents to severe incoming connections from her subordinates. 
Moreover, a second research direction is to study the effect of stability of the equilibrium networks of the game, and further study the conditions that could lead more or less efficient equilibrium networks in terms of the game parameters. Finally, from an organizational point of view, it is relevant to theoretically understand under what conditions are flat or steep hierarchies preferable in terms of the game's parameters.

\bibliographystyle{plainurl}
\bibliography{alias,Main,FB,Excl_this}

\appendix
\section{Basic graph theoretical definitions}
\label{app-basgraph}
Consider a directed graph (digraph) $G$. 
%
Graph $G$ is \emph{strongly connected} whenever there exist a path between any two of its nodes; and, by convention, any graph composed by a single node is in itself strongly connected. We say that $i$ and $j$ are \emph{weakly connected} when, if we consider all edges of $G$ being undirected, there exists a path from one of these nodes to the other, and we say that $G$ is \emph{(weakly) connected} if any two nodes are (weakly) connected. 
If $G$ is not connected, then any connected subgraph 
such that no other node of $G$ is connected to any of its nodes, is called an~\emph{isolated network}. 
%

A \emph{(strongly) connected component} of $G$ is a subgraph of $G$ that is strongly connected and such that no additional edges or nodes from $G$ can be included in the subgraph without violating its property of being strongly connected. We can partition the nodes of $G$ according to which connected component they belong to. The \emph{condensation graph} of $G$ is the graph resulting from contracting every strongly connected component of $G$ into a single node, and it is an acyclic graph. 
%
%
%

\section{Proofs for the Static Analysis section}
\label{proofs-results}

We first introduce a couple of useful technical results.

\begin{lemma}
\label{l1}
For both non-consensual and consensual agents, an equilibrium network cannot have directed cycles, and thus, only has undirected connected components. In particular, for the case of non-consensual agents, all connected components are complete subgraphs. 
\end{lemma}
\begin{proof}
By contradiction, consider an equilibrium network $G$ and a single edge $ij\in G$ 
such that a directed cycle involving the node $i$ exists with the edge $ij$ belonging to such cycle. 
We analyze whether $j$ has an incentive to start a collaboration with $i$. Observe that $H(\ell_j(G))=H(\ell_j(G+ji))$, $|P_j(G,k)|\leq |P_j(G+ji,k)|$ for any $k\in\Nins(j,G)$ with $k\neq i$. 
Then,
\begin{align*}
u_j(G+ji)-u_j(G)&=\left(\gamma(|\Nout(j,G)|+1)+H(\ell_j(G+ji))-c\sum_{k\in\Nins(j,G)\setminus\{i\}}\frac{1}{|P_j(G+ij,k)|}\right)\\
&-\left(\gamma|\Nout(j,G)|+H(\ell_j(G))-c\sum_{k\in\Nins(j,G)}\frac{1}{|P_j(G,k)|}\right)\\
&=\gamma+c\left(\sum_{k\in\Nins(j,G)}\frac{1}{|P_j(G,k)|}-\sum_{k\in\Nins(j,G)\setminus\{i\}}\frac{1}{|P_j(G+ij,k)|}\right)\\
&>0.
\end{align*}
Thus, $j$ has an incentive to establish the edge $ji$. For non-consensual agents, this would contradict the fact that $G$ is an equilibrium network, and so $G$ cannot have single edges in any connected component, which implies that all connected components are undirected subgraphs. Now, assume that agents are consensual and continue the analysis to see if $i$ has an incentive in accepting a link from $j$. We can easily show that 
$$u_i(G+ji)-u_i(G)=c\sum_{k\in\Nins(i,G)}\left(\frac{1}{|P_i(G,k)|}-\frac{1}{|P_i(G+ji,k)|}\right)\geq 0$$ 
since $|P_i(G+ij,k)|\geq |P_i(G,k)|$ for any $k\in\Nins(i,G)$. 
Thus, $i$ will accept collaborating with $j$, and this immediately implies that $G$ must only have undirected connected components for consensual agents.

Now, assume we have any $i,j\in N$ with $ij,ji\notin G$ belonging to the same undirected connected component of $G$ and that this connected component is not a complete subgraph. Then, it is easy to show that $u_i(G+ij)-u_i(G)=\gamma>0$, and so $i$ has an incentive to form the connection $ij$. For non-consensual agents, this contradicts $G$ being an equilibrium network, from which we immediately conclude that all connected components must be complete subgraphs. Now, consider consensual agents. Then, it is easy to show that $u_j(G+ij)-u_j(G)=-c\frac{1}{|P_j(G,i)|}<0$, and so $j$ has no incentive in accepting a single edge from $i$. Thus, the connected components of $G$ are not necessarily complete for consensual agents. 
\end{proof}
\begin{corollary}
\label{l2}
Consider any graph $G$ without directed cycles. Then, no agent has an incentive to sever any single edge.
\end{corollary}
\begin{proof}
Given the graph $G$, consider any $i,j\in N$ such that $ij\in G$ and $ji\notin G$; and observe that $u_i(G-ij)-u_i(G)=-\gamma<0$. 
\end{proof}

Now, we introduce the important proofs of our Solution Analysis section.

\begin{proof}[Proof of Theorem~\ref{th:non-c}]
Throughout the proof, let $T_i(G):=|\Nout(i,G)|\gamma$.  We denote by $i\to
j$ the fact that there exists a path between $i$ and $j$, and, with a
slight abuse of notation, we will also use it to denote the set of edges
that compose some path between $i$ and $j$, so that clearly $i\to j \subset
G$.\footnote{Note that there might be multiple paths between $i$ and $j$
  and so there is some ambiguity as to which specific path $i\to j$ refers
  to, but this will not be a problem throughout the paper.}

\textbf{Part I: }We first prove that any given sequential hierarchy as stipulated in items~\ref{cond-1} and~\ref{cond-2} of the theorem statement, is the equilibrium network for some game $\mathcal{G}$. Consider a given sequential hierarchy $G$. 
Now, we choose any $c,\gamma>0$ and fix their values. Then, the remaining problem is to find a reward function $H$ so that $G$ is an equilibrium network for the game $\mathcal{G}=(N,u,\mathcal{A},\gamma,H,c)$. Then, for $G$ to be an equilibrium network, we need to ensure that no agent is able to establish new links or severe existing ones as pointed out in Definition~\ref{eq_net}. Let $k$ be the number of levels  of $G$.

If $k=1$, then the network is complete with $u_i(G-ij)=\gamma(n-2)+H(0)-c/(n-1)$ and $u_i(G)=\gamma(n-1)+H(0)$ for any $i,j\in N$; which gives $u_i(G-ij)-u_i(G)=-\gamma-c/(n-1)<0$ since $n\geq 3$. Thus, there is no incentive for any agent of the network to severe a link, irrespective of the function $H$. Thus, any non-negative monotonic function $H$ can be chosen as the reward function.

Let us assume now that we have $k>1$. 
%

\emph{Case 1: }Assume $i$ has level $0<\ell\leq k-1$ and that she considers establishing a link with some agent $j$ that has level $\ell-1$ (recall that $ij\in G$), which we call an \emph{immediate backward connection}. Then, $u_i(G)=T_i(G)+H(\ell)-c\frac{|\setdef{j}{\ell_j(C(G))<\ell}|}{|P_i(G)|}$ and $u_i(G+ij)=T_i(G)+\gamma + H(\ell-1)-c\frac{|\setdef{j}{\ell_j(C(G))<\ell}|-1}{|P_i(G)|+1}$, and thus 
\begin{equation*}
u_i(G+ij)-u_i(G)= H(\ell-1)-H(\ell)+\gamma+c\left(\frac{|P_i(G)|+|\setdef{j}{\ell_j(C(G))<\ell}|}{|P_i(G)|(|P_i(G)|+1)}\right).
\end{equation*}
%
%
%
Having this previous expression being less than zero would suffice for a non-consensual agent to not have an incentive to establish immediate backward links, which is the same as choosing a function $H$ that satisfies condition~\ref{cond-1} of the theorem statement (observe that $S$ must be, indeed, strictly increasing since the sequence $\{|\setdef{j}{\ell_j(C(G))<\ell}|\}_{\ell=1}^{k-1}$ is also strictly increasing for $k>2$).

\emph{Case 2: }Assume $i$ has level $0<\ell\leq k-1$ and that she considers establishing a connection with an agent $j$ that has level $\ell-m$ for $m>1$ (assuming, obviously, that such $j$ exists). 
Let $q_\ell(G)=|\setdef{j}{\ell_j(C(G))=\ell}|$ be the number of agents that have rank $\ell$ in the graph $G$. 
Then, $u_i(G)=T_i(G)+H(\ell)-c\frac{|\setdef{j}{\ell_j(C(G))<\ell}|}{q_\ell(G)}$ and $u_i(G+ij)=T_i(G)+\gamma + H(\ell-m)-c\left(\frac{|\setdef{j}{m<\ell_j(C(G))<\ell}|}{q_\ell(G)}+ \frac{|\setdef{j}{\ell_j(C(G))\leq m}|}{q_\ell(G)+1} \right)$; and thus 
\begin{align*}
u_i(G+ij)-u_i(G)&=-(H(\ell)-H(\ell-m))+\gamma+c\frac{|\setdef{j}{\ell_j(C(G))\leq m}|}{q_\ell(G)(q_\ell(G)+1)}\\
=&-\sum_{p=0}^{m-1} (H(\ell-p)-H(\ell-(p+1)))+\gamma+c\frac{|\setdef{j}{\ell_j(C(G))\leq m}|}{q_\ell(G)(q_\ell(G)+1)}.
\end{align*}
%
Now, if~\eqref{cond_im} holds, then $-(H(\ell-p)-H(\ell-(p+1))<-\gamma-c\left(\frac{q_{\ell-p}(G)+|\setdef{j}{\ell_j(C(G))<\ell-p}|}{q_{\ell-p}(G)(q_{\ell-p}(G)+1)}\right)$ for any $0\leq p\leq \ell-1$; from which it follows that
\begin{align*}
u_i(G+ij)-u_i(G)&=-(H(\ell)-H(\ell-1))-\sum_{p=1}^{m-1} (H(\ell-p)-H(\ell-(p+1)))+\gamma+c\frac{|\setdef{j}{\ell_j(C(G))\leq m}|}{q_\ell(G)(q_\ell(G)+1)}\\
&<-\gamma-c\frac{q_{\ell}(G)+|\setdef{j}{\ell_j(C(G))<\ell}|}{q_\ell(G)(q_\ell(G)+1)}-\sum_{p=1}^{m-1} (H(\ell-p)-H(\ell-(p+1)))+\gamma\\
&\quad+c\frac{|\setdef{j}{\ell_j(C(G))\leq m}|}{q_\ell(G)(q_\ell(G)+1)}\\
&<-c\frac{1}{q_\ell(G)+1}-\sum_{p=1}^{m-1} (H(\ell-p)-H(\ell-(p+1)))\\
&<0;
\end{align*}
%
%
and so $i$ has no incentive to establish such new link. In conclusion, if condition~\ref{cond-1} from the theorem statement holds, then no agent has in incentive to establish a link with any agent that has a lower level.

\emph{Case 3: } Assume $i$ has any level $0\leq \ell \leq k-1$ and that she considers severing any of her existing edges. From Corollary~\ref{l2}, we know that she has no incentive to sever any single edge. Now, consider any agent $j$ such that $ij\in G$ belongs to some undirected edge. If the undirected edge is a critical edge, i.e., $|P_{i}(G)|=2$, then it is easy to show that 
$u_i(G-ij)-u_i(G)=-\gamma+H(\ell+1)-H(\ell)-c-c\frac{|\setdef{j}{\ell_j(C(G))<\ell}|}{2}$. Having this previous expression being less than zero would suffice for $i$ to not have an incentive to severe the edge $ij$, which is the same as choosing a function $H$ that satisfies condition~\ref{cond-2} from the theorem statement. On the other hand, if the undirected edge is not a critical edge, so that $|P_i(G)|>2$, then it is easy to show that $u_i(G-ij)-u_i(G)=-\gamma-c\frac{|\setdef{j}{\ell_j(C(G))<\ell}|+|P_i(G)|}{|P_i(G)|(|P_i(G)|-1)}<0$; and thus, $i$ has no incentive in severing the edge $ij$.

In conclusion, from all the previous analyzed cases, if $G$ is a sequential hierarchy with more than one level, then we only need to define the non-negative 
monotonic function $H$ as stipulated by the conditions~\ref{cond-1} and~\ref{cond-2} of the theorem statement in order to finally obtain the sought game $\mathcal{G}$ such that $G$ is an equilibrium network for $\mathcal{G}$.
%
%
%

\textbf{Part II: }Now, we need to prove that any equilibrium network is a sequential hierarchy as stipulated in items~\ref{cond-1} and~\ref{cond-2} of the theorem statement. Consider the given equilibrium network $G$. From Lemma~\ref{l1} we immediately conclude that $G$ is a hierarchical structure with complete connected components. Note that this lemma does not specify that $G$ must be connected. 
However, from the proof of the same lemma, if we assume $G$ is not connected, then there is always an incentive for some agent to form a single edge with an agent from another complete connected component of the same level, which by contradiction implies that there is only one connected component per any level in $G$. 
Now, we claim that any agent in $G$ must have single edges towards any other agent with a higher rank than her. Assume this is not the case and take any $i,j\in N$ such that $ij\notin G$ and $\ell_i(C(G))<\ell_j(C(G))$, 
then it follows easily that $i$ has an incentive to form links with $j$ since such action would increase her utility by $\gamma>0$, giving a contradiction. Thus, in observance of Definition~\ref{def1}, we have proved that $G$ is a sequential hierarchy. It should be noted that, from our previous analysis in Part I, the number of agents in each complete connected component of $G$ (when there is more than one level) is determined by the utility parameters $(\gamma,H,c)$ satisfying the conditions~\ref{cond-1} and~\ref{cond-2} of the theorem statement. This concludes the proof that any equilibrium network is a sequential hierarchy as stipulated in items~\ref{cond-1} and~\ref{cond-2} of the theorem statement.

Finally, assume we have a sequential hierarchy $G$ such that there exists some level $\ell>1$ such that: 1) for any $i\in N$ that has level $\ell$, we have $|\setdef{j}{\ell_j(C(G))<\ell}|\geq \frac{2|P_i(G)|^2}{2+|P_i(G)|+|P_i(G)|^2}$; and 2) for any $p\in N$ that has level $\ell-1$, we have $|P_p(G)|=2$. Then, we claim that $G$ 
cannot be an equilibrium network for any game. To see this, observe condition 1) implies $\left(\frac{|P_i(G)|+|\setdef{j}{\ell_j(C(G))<\ell}|}{|P_i(G)|(|P_i(G)|+1)}\right)\geq 1+\frac{|\setdef{j}{\ell_j(C(G))<\ell}|}{2}$, which immediately implies both~\eqref{cond_im} and~\eqref{cond_2} cannot hold simultaneously because of condition 2), for any utility parameters $(\gamma,H,c)$.
\end{proof}

\begin{proof}[Proof of Theorem~\ref{th:c}]
Throughout the proof, let $T_i(G):=|\Nout(i,G)|\gamma$. We adopt the same notation introduced at the beginning of the proof of Theorem~\ref{th:non-c}, and we additionally introduce the following notation: given nodes $i,j\in N$, let $i\notto j$ denote the fact that there does not exist any path between $i$ and $j$ in $G$, and let $i-j$ denote that $i$ and $j$ are (weakly) connected, i.e., $i\to j$ or $j \to i$. 

\textbf{Part I: }We first analyze the conditions under which any given hierarchical structure, as in items~\ref{lab_1},~\ref{lab_2} and~\ref{cond-2a} of the theorem statement,  
is the equilibrium network for some game $\mathcal{G}$. Consider a given network $G$ that has one or more hierarchical structures. 
The remaining problem is to find the utility parameters $(\gamma,H,c)$ so that $G$ is an equilibrium network for the game $\mathcal{G}=(N,u,\mathcal{A},\gamma,H,c)$. 
%
Now, we analyze the conditions that ensure no agent has an incentive and/or is able to form a new link or severe any existing link, so that $G$ can be properly defined as an equilibrium network of $\mathcal{G}$.

\emph{Case 1:} Assume any agent $i$ that has level $0<\ell\leq k-1$, with $k>1$ being the number of levels in the hierarchical structure for which $i$ belongs to, is considering establishing a link with any agent $j$ that has rank $\ell-m$ for any appropriate $m\geq 1$, i.e., establish a backward connection. First, assume $ji\in G$. 
Then, following a similar procedure as for Case 1 of Part I in the proof of Theorem~\ref{th:non-c}, we obtain $u_i(G+ij)-u_i(G)=\gamma+H(\ell-m)-H(\ell)+c\sum_{p\in\Nins(i,G)\setminus\{j\}}\left(\frac{1}{|P_i(G,p)|}-\frac{1}{|P_i(G+ij,p)|}\right)+\frac{c}{|P_i(G,j)|}$ with $|P_i(G+ij,p)|\geq |P_i(G,p)|$ for any $p\in\Nins(i,G)\setminus\{j\}$, so that $i$ does not have an incentive to establish a link with $j$ if the inequality~\eqref{cond_im1} from condition~\ref{lab_1} of the theorem statement is satisfied. Assume $i$ has an incentive to establish this link with $j$, then since agents are consensual, we need to analyze if $j$ would consent to such action from $i$. Then, it is easy to show that  
$u_j(G+ij)-u_j(G)=c\sum_{p\in\Nins(j,G)}\left(\frac{1}{|P_j(G,p)|}-\frac{1}{|P_j(G+ij,p)|}\right)\geq  0$ since $|P_i(G+ij,p)|\geq |P_i(G,p)|$ for any $p\in\Nins(j,G)$, and thus $j$ would accept such proposition. Then, the sufficient and necessary condition for not incentivizing a backward connection from $i$ is the satisfaction of condition~\ref{lab_1} 
from the theorem statement. 
Second, assume $ji\notin G$ and $j\to i$. 
Then, it is easy to see that $u_i(G+ij)-u_i(G)=\gamma+H(\ell-m)-H(\ell)$. Assume the worst case, i.e., 
$i$ has an incentive in making such connection. Then, we find that $u_j(G+ij)-u_j(G)=-c<0$, and so $j$ will not accept a new link from $i$. 
Third, assume $ji\notin G$ and $j\notto i$ 
(notice that $i$ and $j$ may or may not belong to the same hierarchical structure). Then, it is easy to show that $i$ has an incentive in establishing such link, and that for $j$ we have $u_j(G+ij)-u_j(G)=H(\ell+1)-H(\ell-m)-c$, and so $j$ will not accept such connection if and only if 
\begin{equation}
\label{la}
H(\ell+1)-H(\ell-m)<c.
\end{equation}

%

\emph{Case 2: }Assume any agent $i$ that has level $0\leq \ell< k-1$, with $k>1$ being the number of levels in the hierarchical structure for which $i$ belongs to, considers making a forward connection with any agent $j$ that has rank $\ell+m$ for any appropriate $m\geq 1$. We already know from Part I in the proof of Theorem~\ref{th:non-c} that $i$ has an incentive to establish a link with $j$. 
If $m=1$ and $i\to j$, then it follows that $u_j(G+ij)-u_j(G)=-\frac{c}{|P_j(G,i)|}<0$. If $m=1$ and $i\notto j$, or $m>1$ (whether $i-j$ or not), then $u_j(G+ij)-u_j(G)=-c<0$. 
Then, we conclude that $j$ has no incentive in accepting a connection from $i$. 
%

\emph{Case 3:} Assume any agent $i$ that has level $0\leq \ell\leq k-1$, with $k\geq 1$ being the number of levels in the hierarchical structure for which $i$ belongs to, considers making a connection with any agent $j$ of rank $\ell$ (so that $ij,ji\notin G$). Then, we already know $i$ has an incentive for making such connection. If $j\in\Co(i,G)$, then we have $u_j(G+ij)-u_j(G)=-\frac{c}{|P_j(G,i)|}<0$ and thus $j$ will not accept such link from $i$. 
If $j\notin\Co(i,G)$ (whether $i$ belongs or not to the same hierarchical structure as $j$), then we obtain 
$u_j(G+ij)-u_j(G)=H(\ell+1)-H(\ell)-c$ and so $j$ has no incentive in accepting such link from $i$ if and only if 
\begin{equation}
\label{lc}
H(\ell+1)-H(\ell)<c.
\end{equation}

\emph{Case 4: }
Assume any agent $i$ that has any rank $\ell$ considers severing any of her existing edges. From Corollary~\ref{l2}, we know that she has no incentive to sever any single edge. Now, consider any agent $j$ such that $ij\in G$ belongs to some undirected edge. If the undirected edge is a critical edge and $i$ can increase its level by severing the edge $ij$, then it is easy to show that 
$u_i(G-ij)-u_i(G)=-\gamma+H(\ell+1)-H(\ell)-c\sum_{p\in\Nins(i,G)}\left(\frac{1}{|P_i(G-ij,p)|}-\frac{1}{|P_i(G,p)|}\right)-c$. Then, the condition $H(\ell+1)-H(\ell)<\gamma+c\sum_{p\in\Nins(i,G)}\left(\frac{1}{|P_i(G-ij,p)|}-\frac{1}{|P_i(G,p)|}\right)+c$ would suffice for $i$ to not have an incentive to sever the edge $ij$, which is expressed in equation~\eqref{cond_2a} from condition~\ref{cond-2a} of the theorem statement. On the other hand, if the undirected edge is not a critical edge, then $|\Co(i,G)|> 2$ and it is easy to show that $u_i(G-ij)-u_i(G)=-\gamma-c\sum_{p\in\Nins(i,G)}\left(\frac{1}{|P_i(G-ij,p)|}-\frac{1}{|P_i(G,p)|}\right)-c<0$ since $|P_i(G-ij,p)|\leq|P_i(G,p)|$ for any $p\in\Nins(i,G)$; and thus, $i$ has no incentive in severing the edge $ij$.

Notice that enforcing inequalities~\eqref{la} and \eqref{lc} amounts to the condition~\ref{lab_2} of the theorem statement. Therefore, if there exist utility parameters $(\gamma,H,c)$ that satisfy the conditions~\ref{lab_1},~\ref{lab_2} and \ref{cond-2a} from the theorem statement, we conclude that the given sequential hierarchy is an equilibrium network. 

\textbf{Part II: }Now, consider a given equilibrium network $G$. From Lemma~\ref{l1} we immediately conclude that $G$ is a hierarchical structure or is composed of isolated hierarchical structures (since this lemma does not specify whether $G$ is connected or not). 
%
Then, from our previous analysis in Part I, it follows that the utility parameters $(\gamma,H,c)$ must be such that the conditions~\ref{lab_1},~\ref{lab_2} and~\ref{cond-2a} in the theorem statement are satisfied. This concludes the proof that any equilibrium network is composed by one or multiple hierarchical structures. 
%
%
%

Finally, we prove the existence of hierarchical structures that are not equilibrium networks of any possible hierarchical network formation game. For the case in which the hierarchical structures are sequential hierarchies, we can use the same construction given in the proof of Theorem~\ref{th:non-c}. Now, we focus on the case in which the hierarchical structures are not sequential hierarchies. Consider a hierarchical structure $G$ with more than one component per any level and with complete connected components, and such that the only single edges in the graph are as follows: for any $i,j\in N$ with $\ell_j(C(G))=\ell_j(C(G))-1\geq 0$, it holds 
$pk\in G$ (and so $kp\notin G$) for any $p\in\Co(i,G)$, $k\in\Co(j,G)$. Moreover, also assume $|P_i(G)|>2$ for any $i\in N$, i.e., there are no critical edges in $G$. 
We immediately notice from the topology of $G$ and our discussion above that the expression in equation~\eqref{cond_im1} from condition~\ref{lab_1} of the theorem statement becomes 
%
%
\begin{equation}
\label{simple}
H(\ell)-H(\ell-1)>\gamma+c\left(\frac{|P_i(G)|+|P_j(G)|}{|P_i(G)|(|P_i(G)|+1)}\right)
\end{equation}
for any $i,j\in N$ with $\ell_i(C(G))=\ell$ and $\ell_j(C(G))=\ell-1$, $\ell>0$.  
Assume $G$ has more than one hierarchical structure, and that one of them has $\ell+1>0$ levels with $i_p\in N$ denoting any agent of such hierarchical structure such that $\ell_i(C(G))=p\in[0,\ell-1]$. Consider any level $0<\ell\leq k-1$. Now, observe that condition~\ref{lab_1} (using~\eqref{simple}) and condition~\ref{lab_2} imply 
\begin{align}
H(\ell)-H(0)&>\ell\gamma+c\sum_{p=1}^\ell\left(\frac{|P_{i_p}(G)|+|P_{i_{p-1}}(G)|}{|P_{i_p}(G)|(|P_{i_p}(G)|+1)}\right),\label{aa1}\\
H(\ell)-H(0)&<c,\label{aa2}
\end{align}
respectively. 
We claim that if there exists some level $\ell$ such that $1\leq \sum_{p=1}^\ell\left(\frac{|P_{i_p}(G)|+|P_{i_{p-1}}(G)|}{|P_{i_p}(G)|(|P_{i_p}(G)|+1)}\right)$, 
then there are no utility parameters that can define a game such that $G$ is its equilibrium network. To see this, observe that this condition implies $c<\ell\gamma+c\sum_{p=1}^\ell\left(\frac{|P_{i_p}(G)|+|P_{i_{p-1}}(G)|}{|P_{i_p}(G)|(|P_{i_p}(G)|+1)}\right)$ for any $\gamma,c>0$, which makes it impossible to choose a reward function $H$ such that conditions~\eqref{aa1} and~\eqref{aa2} are satisfied, which in turn violates both conditions~\ref{lab_1} and~\ref{lab_2} from the theorem. 
\end{proof}

\begin{proof}[Proof of Proposition~\ref{lem-ex}]
We analyze the conditions under which any given network as in~\ref{s11} or~\ref{s22} is the equilibrium network for some game $\mathcal{G}$ with consensual agents. Given the network, the remaining problem is to find the appropriate utility parameters $(\gamma,H,c)$ so that it becomes the equilibrium network for the game $\mathcal{G}=(N,u,\mathcal{A},\gamma,H,c)$. 

Consider $G$ to be a given hierarchical structure as in statement~\ref{s11} with $k>1$ levels. First, observe that since all connected components are complete subgraphs with more than two nodes, then there are no critical edges in $G$, so that we do not need to analyze the satisfaction of condition~\ref{cond-2a} of Theorem~\ref{th:c}. 
%
%
We immediately notice from the structure of $G$ and the proof of Theorem~\ref{th:c} that equation~\eqref{cond_im2} of condition~\ref{lab_2} from Theorem~\ref{th:c} becomes 
\begin{equation}
\label{simple2}
H(\ell)-H(\ell-1)>\gamma+c\left(\frac{|P_i(G)|+|P_j(G)|}{|P_i(G)|(|P_i(G)|+1)}\right)
\end{equation}
with $\ell>0$ and any $i,j\in N$ with $\ell_i(C(G))=\ell>0$ and $\ell_j(C(G))=\ell-1$. Then, it follows immediately that we can first choose any $\gamma,c>0$ and then an appropriate non-negative monotonic function $H$ satisfying equation~\eqref{simple2} to form the sought game. 

Now, consider $G$ to be a given network as in statement~\ref{s22}. 
Observe that all isolated networks in $G$ have more than one level. First, we make the observation that, considering any agent $i,j\in N$ with $\ell_i(C(G))=\ell>0$ and $\ell_j(C(G))=\ell-1$, then, from the conditions of $G$, 
the following inequalities hold: 
\begin{equation}
\label{eq-11}
\frac{|P_i(G)|+|P_j(G)|}{|P_i(G)|(|P_i(G)|+1)}\leq \frac{2}{|P_i(G)|+1}<1,
\end{equation}
since $|P_i(G)|>2$.
%

Using the inequality in~\eqref{simple2}, we conclude that if there exists a tuple of utility parameters $(\gamma,H,c)$ such that they satisfy the condition 
\begin{equation}
\label{reff}
\gamma+c U(G)<H(\ell)-H(\ell-1)< c \,\text{ with }\, U(G)=\max_{\substack{i\,:\,\ell_i(C(G))=\ell\\j\,:\,\ell_j(C(G))=\ell-1}}\left(\frac{|P_i(G)|+|P_j(G)|}{|P_i(G)|(|P_i(G)|+1)}\right),
\end{equation}
then this will imply the satisfaction of conditions~\ref{lab_1} and~\ref{lab_2} from Theorem~\ref{th:c} and thus create a game such that $G$ is its equilibrium network. We claim that such tuple exist. To see this, notice that~\eqref{eq-11} implies 
$cU(G)< c$, 
and so we can choose a sufficiently small $\gamma>0$ so that 
$\gamma+cU(G)< c$, 
and therefore we can squeeze in the difference $H(\ell)-H(\ell-1)>0$ (which immediately defines the reward function $H$ appropriately) so that~\eqref{reff} holds.
\end{proof}

\section{Proofs for the Dynamic Analysis section}
\label{proofs-results1}


\begin{proof}[Proof of Theorem~\ref{th_conv}]
%

For any type of agents, let $\mathcal{E}$ be the set of equilibrium networks. Consider a set of graphs $\mathcal{S}=\{G_1,\dots,G_k\}$ (with, obviously, different elements). We say $\mathcal{S}$ is an \emph{invariant set} in our network formation process if the following holds: for any elements $G_i,G_j\in\mathcal{S}$, if $G(t)\in G_i$ then $G(t')=G_j$ for infinitely many $t'\geq t$ with probability one. Notice that this implies that, if $G_i\in\mathcal{S}$ and $G_k\notin\mathcal{S}$, then there is zero probability that the network formation process starting at $G_i$ can eventually be the network $G_k$. In particular, an invariant set $\mathcal{S}$ has only one element if and only if this element is an invariant network (see Definition~\ref{def_inv_net}). 
Finally, we remark that in the literature of Markov chains, this concept of invariant set defined above is equivalent to the concept of an irreducible closed set of persistent states~\cite{GG-DR:01}.

Now, according to the Decomposition Theorem of Markov Chains, the state space of the Markov chain $\mathcal{M}$ can be partitioned uniquely as the union of a set of transient states (if such set exists)  and one or more invariant sets. Now, from the fact that $G(t)$ can only take a finite number of different topologies, or, equivalently, from the fact that the Markov chain $\mathcal{M}$ is finite, it is known that, with probability one, the network $G(t)$ 
enters an invariant set in finite time (indeed, each graph in this invariant set is a non-null persistent state~\cite{GG-DR:01}). Obviously, if the network formation process starts in an equilibrium network, we know from Proposition~\ref{prop1} that this finite time is actually zero.
%
%
%
In conclusion, from Proposition~\ref{prop1}, if any invariant set $\mathcal{S}$ has only one element, 
then we conclude that 
$G(t)$ converges to an equilibrium network with probability one. 
%
%
%

For the rest of the proof, pick any type of agents for the game. 
Now, given a graph $G$, define the set $\bar{G}$ as follows: $\bar{G}=\setdef{\{i,j\}}{ij\in G\text{ or }ji\in G}$. 

Now, we make a couple of observations:
\begin{enumerate}[label=(\alph*)]
\item Consider a single edge $ij\in G(t)$ that belongs to a directed cycle at any time $t$, and that $i$ has level $\ell$. Then, it is easy to see that $H(\ell+k)$ is the hierarchical reward of $i$ for the graph $G(t)-ij$, for some $k\geq 0$. Now, observe that $u_i(G(t)-ij)-u_i(G(t))=-\gamma+H(\ell+k)-H(\ell)$. If $k=0$, then $G(t+1)=G(t)$; otherwise, if $k>0$, then, from the conditions on the increments of the reward function in the theorem statement, we get that $u_i(G(t)-ij)>u_i(G(t))$ and so $G(t+1)=G(t)-ij$.
\item \label{case-2a}From the proof of Lemma~\ref{l1}, we conclude that: if $ij$ belongs to a directed cycle in $G(t)$, then, with positive probability, $ij,ji\in G(t+1)$.
\end{enumerate}
Let $G(t_o)$ be the graph at any time $t_o\geq 0$. Then, with these two observations, we can construct a finite sequence of pairs of agents such that they modify the network according to the following steps (a step can be skipped if not applicable):
\begin{enumerate}[label=\arabic*)]
\item \label{oink0}eventually makes the graph connected at some time $t_1\geq t_o$, i.e., $G(t_1)$ has no isolated networks (this is easy to see for non-consensual agents, and for consensual agents, it follows from the fact that inequality~\eqref{cond_im2} in item~\ref{lab_2} of Theorem~\ref{th:c} cannot be satisfied under the conditions on the increments of the reward function $H$),
\item eventually makes the graph to have no directed cycles at some time $t_2\geq t_1$,
\item \label{oink01}taking $G(t_2)$, it starts establishing one or multiple single edges that may result in the formation of new directed cycles (if any), and then immediately turn these edges into undirected ones 
according to~\ref{case-2a} (for example, for non-consensual agents, new directed edges may appear when an agent $i$ forms a single edge with another agent $j$ of the same level and such that there is a path from $j$ to $i$); and repeat this process until the graph has no directed cycles and cannot form new directed cycles at some time $t_3\geq t_2$; notice that $G(t_3)$ is a hierarchical structure,
\item\label{oink1} severs any existing critical edge in $G(t_3)$ (i.e., turns it into a single edge), and continues doing this process until no new critical edge is found in the graph at some time $t_4\geq t_3$ (this can always be done since the condition on the increments of the reward function $H$ implies that the inequality~\eqref{cond_2a} is always not satisfied, which, from the proof of Theorem~\ref{th:c}, let us conclude that any existing critical edge can always be severed).
\end{enumerate}
Notice that the graph $G(t_4)$ is still a hierarchical structure, since severing critical edges (as in step~\ref{oink1}) does not introduce directed cycles, nor introduce new isolated networks. 
%
%
Since $G(t_4)$ is a hierarchical structure and the increments of the reward function $H$ are large enough, there is no incentive for any agent to establish an edge with any agent with a lower level. Notice that for non-consensual agents, step~\ref{oink01} implies all agents with the same level form a complete subgraph. Notice also that for consensual agents, if a single edge is established between two agents $i$ and $j$ with the same level in $G(t_4)$, then, for any $t\geq t_4$, this single edge will not be an undirected edge with probability one, nor no new directed cycle that contains $i$ and $j$ will be formed with probability one. 
Therefore, we conclude that, for any type of agents, no new undirected edges nor new directed cycles can be formed for any $G(t)\geq G(t_4)$ with probability one. 

So, given the arbitrary initial network $G_o$, let us define $T$ to be the random time at which $G(t)$ becomes a hierarchical structure as characterized by $G(t_4)$. Then, it is easy to observe that $T$ is a stopping time and we can use the first Borel-Cantelli lemma (as in~\cite[proof of Theorem~3.1]{PCV-KSC-FB:18p}) to prove that $T<\infty$ with probability one, since any finite sequence of pairs of agents that implements the appropriate changes in the steps~\ref{oink0} to~\ref{oink1} (if any) always exists with positive probability at any time $t$. Then, we can use the strong Markov property, and \emph{assume} that the network formation process starts at $T$ with the graph $G(T)$. Since only new single edges may appear at future time-steps and these single edges cannot form directed cycles, we conclude that, with probability one, $\bar{G}(t+1)\supseteq\bar{G}(t)$ for any $t\geq T$, i.e., $|\bar{G}(t+1)|\geq|\bar{G}(t)|$. Moreover, for any $t\geq T$, we conclude that $G(t)$ has a pair of agents that can establish a single edge at time $t+1$ if and only if $|\bar{G}(t+1)|>|\bar{G}(t)|$ with positive probability. Now, consider any invariant set $\mathcal{S}$. It follow from these observations that if $G(t)\in\mathcal{S}$, then $|\bar{G}(t+1)|=|\bar{G}(t)|$ with probability one. However, this only occurs if and only if no new single edges can be added, i.e., when $|\mathcal{S}|=1$. Then, we conclude that any invariant set is an invariant network and so that the network converges to an equilibrium network.

\end{proof}

\end{document}